\documentclass[journal]{IEEEtran}
\usepackage[hidelinks]{hyperref}
\usepackage{amsmath,amsfonts,amssymb}
\usepackage{xcolor}
\usepackage{graphicx}
\usepackage{float}
\usepackage{amsthm}
\usepackage{blkarray}
\usepackage{multirow}

\theoremstyle{definition}
\newtheorem{definition}{Definition}[section]
\newtheorem{theorem}{Theorem}[section]
\newtheorem{lemma}{Lemma}[section]
\newtheorem{proposition}{Proposition}[section]
\newtheorem{corollary}{Corollary}[section]

\newtheorem{example}{Example}[section]
\numberwithin{equation}{section}

\newcommand{\D}{\ensuremath{q}} 

\newcommand{\ket}[1]{|{#1}\rangle}
\newcommand{\bra}[1]{\langle{#1}|}
\DeclareMathOperator{\Tr}{Tr}
\renewcommand{\vec}[1]{\mathbf{#1}}

\begin{document}

\title{Quantum weight enumerators and tensor networks}

\author{ChunJun~Cao\thanks{C. Cao is at the Institute for Quantum Information and Matter, California Institute of Technology.} 
and Brad~Lackey\thanks{B. Lackey is at Microsoft Quantum, Microsoft Corporation.}}

\maketitle

\begin{abstract}
We examine the use of weight enumerators for analyzing tensor network constructions, and specifically the quantum lego framework recently introduced. We extend the notion of quantum weight enumerators to so-called tensor enumerators, and prove that the trace operation on tensor networks is compatible with a trace operation on tensor enumerators. This allows us to compute  quantum weight enumerators of larger codes such as the ones constructed through tensor network methods more efficiently. We also provide a general framework for quantum MacWilliams identities that includes tensor enumerators as a special case.
\end{abstract}
\begin{IEEEkeywords}
Quantum codes, weight enumerators, tensor networks, MacWilliams identity
\end{IEEEkeywords}

\noindent{This work has been submitted to the IEEE for possible publication. Copyright may be transferred without notice, after which this version may no longer be accessible.}

\tableofcontents

\section{Introduction}

The weight enumerator of a classical code counts the number of codewords at each Hamming weight, typically presented as a polynomial in a single indeterminate $z$ where the coefficient of $z^w$ is the number of codewords of Hamming weight $w$. The distance of the code can be extracted as the smallest (nonzero) degree where the enumerator has a positive coefficient. The MacWilliams identity \cite{macwilliams1963theorem} famously relates the weight enumerator of a code with that of its dual code.

Shor and Laflamme introduced two quantum weight enumerators \cite{shor1997quantum} related by an analogue of the MacWilliams identity. For stabilizer codes, the $A$-enumerator counts the number of stabilizers at each Hamming weight (i.e. the number of non-identity Pauli operators) while the $B$-enumerator counts the number of logical operators at each Hamming weight. The distance of a stabilizer code is then the smallest degree where the coefficient (up to appropriate normalization) of $B$ strictly exceeds that of $A$. For completeness, we will recount this in Section \ref{section:scalar} below. Note that in terms of error detection, the difference of $B$ and $A$ expresses the probability of an undetected error even when the underlying code is not a stabilizer code \cite{ashikhmin2000quantum}.

Other quantum enumerators have been introduced, such as the ``shadow'' enumerator of \cite{rains1998quantum, rains1998shadow} to create bounds on the length, dimension, and distance for which quantum codes can exist. In \cite{hu2019complete, hu2020weight}, the authors introduce quantum enumerators for asymmetric errors that count the types of Pauli operators that appear in a stabilizer/logical operator. 

In this work we introduce a quantum weight enumerator formalism for tensor network constructions of quantum codes and in particular those from the recent quantum lego formalism \cite{cao2022quantum}\footnote{See also a related approach \cite{Farrelly}.}. We will also examine its detailed applications to quantum error correcting codes using tensor networks in a companion paper \cite{followup}. In particular, the enumerators of Shor and Laflamme (or Hu, et al.) are ``scalar'' enumerators in our theory. We introduce vector and tensor enumerators whose coefficients are scalar enumerators, and show that tracing these enumerators in a tensor network recovers the Shor and Laflamme enumerators. We prove an analogue of the MacWilliams identity for tensor enumerators (Theorem \ref{theorem:tensor-macwilliams-identity}).

We also prove tensor enumerators are well-behaved under tensor product of the underlying networks: the tensor enumerator of a tensor product is the tensor product of the component enumerators (Proposition \ref{proposition:tensor-homomorphism}). We also show that the trace of two legs in a tensor network extends to tensor enumerators, in that the enumerator of the trace can be computed by performing the trace of tensor enumerator along the corresponding legs (Theorem \ref{theorem:enumerator-trace}). In this way we can compute enumerators of large quantum codes constructed through tensor network methods.

\section{Definitions}

A unitary basis on a Hilbert space $\mathfrak{H}$ of $\dim(\mathfrak{H}) = \D$ is a set of unitary operators $\mathcal{E}$ containing $I$ that form an orthonormal basis of $L(\mathfrak{H})$, in that $\Tr(EF^\dagger) = 0$ if $E \not= F$  \cite{werner2001all}. While quantum weight enumerators can be defined relative to a general unitary basis, to obtain important results such as a MacWilliams identity we require the basis have some additional structure. We focus on a class of ``nice error bases'' \cite{knill1996group, knill1996non}, those with Abelian index group \cite{klappenecker2002beyond}. To avoid unnecessary repetition, we will simply refer to these as an ``error basis'' formally defined as follows.

\begin{definition}
    An \emph{error basis} on a Hilbert space is a unitary basis $\mathcal{E}$ that additionally satisfies
    \begin{equation}\label{equation:error-basis}
        EF = \omega(E,F) FE
    \end{equation}
    for all $E,F \in \mathcal{E}$, where $\omega(E,F) \in \mathbb{C}$ satisfies $\omega(E,F)^r = 1$ for some fixed $r$.
\end{definition}

Perhaps the most important example of such a basis in the (generalized) Pauli group. This group is generated by operators $X,Z$, where
$$Z\ket{j} = \zeta^j\ket{j} \text{ and } X\ket{j} = \ket{j+1 \text{ (mod $\D$)}},$$
where $\zeta = e^{2\pi i/\D}$. That is, the eigenvalues of $Z$ are the $\D$-th roots-of-unity, and on its eigenbasis $X$ acts as cyclic shift. Then $XZ = \zeta^{-1} ZX$. The Pauli basis is
$$\mathcal{P} = \{ X^a Z^{a'} \::\: a,a' \in \{0,\dots,\D-1\} \},$$
and $\omega(X^a Z^{a'}, X^b Z^{b'}) = \zeta^{(a'b - ab')}$.

Given an error basis for $\mathfrak{H}$, we obtain one for $\mathfrak{H}^{\otimes n}$ via
$$\mathcal{E}^n = \{ E_1 \otimes \cdots \otimes E_n \::\: \text{ each } E_j \in \mathcal{E}\}.$$
Equation (\ref{equation:error-basis}) holds where we define $\omega$ on $\mathcal{E}^n\times\mathcal{E}^n$ to be the product of the $\omega$ on each component. In particular, for the Pauli group $\omega(X^{\vec{a}} Z^{\vec{a}'}, X^{\vec{b}} Z^{\vec{b}'}) = \zeta^{\vec{a}'\cdot\vec{b} - \vec{a}\cdot\vec{b}'}$, where $X^{\vec{a}} = X^{a_1}\otimes \cdots \otimes X^{a_n}$ and similarly for the the other terms.

The weight of such an operator is the number of non-identity tensor factors it contains, $\mathrm{wt}(E_1 \otimes \cdots \otimes E_n) = |\{r \::\: E_r \not= I \}|$. Write $\mathcal{E}^n[d]$ for the set of weight $d$ operators in the error basis.

Given an error basis $\mathcal{E}$ on a Hilbert space $\mathcal{H}$, Shor and Laflamme \cite{shor1997quantum} define two quantum weight enumerators\footnote{We opt for the normalization in \cite{rains1998quantum} as opposed to that in \cite{shor1997quantum}, and take a different convention in conjugate linearity.} for a pair of Hermitian operators $M_1$ and $M_2$, 
\begin{align}
    \label{equation:scalar-A-weight} A_d(M_1,M_2) &= \sum_{E \in \mathcal{E}^n[d]} \Tr(E^\dagger M_1)\Tr(EM_2),\\
    \label{equation:scalar-B-weight} B_d(M_1,M_2) &= \sum_{E \in \mathcal{E}^n[d]} \Tr(E^\dagger M_1 E M_2).
\end{align}
We will often drop $M_1, M_2$ from this notation when they are understood. While it is not obvious from this definition, both these values are invariant under local unitary transformations. That is, if $U = U_1\otimes \cdots \otimes U_n$ is any local unitary then 
$$A_d(U M_1 U^\dagger, U M_2 U^\dagger) = A_d(M_1, M_2)$$ 
and similarly for $B_d$. For completeness, we provide a proof in \S{\ref{section:scalar}} but will show this in a broader context in later sections.

\begin{definition}
    The \emph{weight enumerator polynomials} $A$ and $B$ are
    \begin{align*}
        A(z) &= \sum_{d=0}^n A_d z^d = \sum_{E\in\mathcal{E}^n}  \Tr(E^\dagger M_1)\Tr(E M_2) z^{\mathrm{wt}(E)}\\
        B(z) &= \sum_{d=0}^n B_d z^d = \sum_{E\in\mathcal{E}^n}  \Tr(E^\dagger M_1 E M_2) z^{\mathrm{wt}(E)}.
    \end{align*}
\end{definition}

We will often find it useful to consider the enumerator polynomials in homogeneous form: 
$$A(w,z) = w^n A(z/w) = \sum_{d=0}^n A_d(M_1,M_2) w^{n-d}z^d,$$
and similar for $B$. These two enumerators are not independent, being linked by the quantum MacWilliams identity \cite{shor1997quantum, rains1998quantum}. In the case of $\mathfrak{H} = (\mathbb{C}^\D)^{\otimes n}$ this reads:
$$B(w,z) = A\left(\tfrac{w + (\D^2-1)z}{\D}, \tfrac{w-z}{\D}\right).$$
We will provide a simple proof of this in \S{\ref{section:macwilliams} below, as a consequence of a more general result. 

For classical codes with an alphabet size greater than two, a so-called complete enumerator may be constructed that also has desirable properties such as a MacWilliams identity \cite{macwilliams1963theorem, macwilliams1977theory}. Hu, Yang, and Yau \cite{hu2019complete} noted that quantum codes can be viewed in this light. For each $R \in \mathcal{E}$ introduce a weight function that counts the number of occurrences of $R$ as a factor, $\mathrm{wt}_R(E_1 \otimes \cdots \otimes E_n) = |\{r \::\: E_r = R \}|$ and an indeterminate $u_R$. The complete enumerator polynomials are
\begin{align*}
    E(z) &= \sum_{G\in\mathcal{E}^n}  \Tr(G^\dagger M_1)\Tr(G M_2) \prod_{R\in\mathcal{E}} u_R^{\mathrm{wt}_R(G)}\\
    F(z) &= \sum_{G\in\mathcal{E}^n}  \Tr(G^\dagger M_1 G M_2) \prod_{R\in\mathcal{E}} u_R^{\mathrm{wt}_R(G)}.
\end{align*}
Note that these polynomials are homogeneous of degree $n$. They are related to the usual Shor-Laflamme enumerators by $A(w,z) = E(w,z,\dots,z)$ and $B(w,z) = F(w,z,\dots,z)$.

In the case of the Pauli basis $\mathcal{P}$ of local dimension $\D$, Hu, Yang, and Yau also introduce a double enumerator \cite{hu2020weight}. We define $X$ and $Z$ weights by
\begin{align*}
    \mathrm{wt}^X(P_1 \otimes \cdots \otimes P_n) &= |\{j \::\: P_j = X^a Z^{a'} \text{ where } a \not= 0 \}|\\
    \mathrm{wt}^Z(P_1 \otimes \cdots \otimes P_n) &= |\{j \::\: P_j = X^a Z^{a'} \text{ where } a' \not= 0 \}|.
\end{align*}
For the Pauli group on $\mathbb{C}^2$ we have $\mathrm{wt}^X(\sigma_x) = \mathrm{wt}^X(\sigma_y) = 1$ and $\mathrm{wt}^Z(\sigma_y) = \mathrm{wt}^Z(\sigma_z) = 1$ (other values being zero). The double enumerator polynomials are
\begin{align*}
    C(x,z) &= \sum_{P\in\mathcal{P}^n}  \Tr(P^\dagger M_1)\Tr(P M_2) x^{\mathrm{wt}^X(P)}z^{\mathrm{wt}^Z(P)}\\
    D(x,z) &= \sum_{P\in\mathcal{P}^n}  \Tr(P^\dagger M_1 P M_2) x^{\mathrm{wt}^X(P)}z^{\mathrm{wt}^Z(P)}.
\end{align*}
To homogenize the double enumerators, we need to introduce two variables $w,y$ and define $C(w,x,y,z) = w^n y^n C(x/y,z/w)$ and similarly for $D$. That is, we homogenize in bidegree $(n,n)$ where $y$ is the homogenizing variable for $x$ and $w$ is that for $z$.

Both the double and complete enumerators are connected with versions of the quantum MacWilliams identity \cite{hu2020weight}. Like the MacWilliams identity for the Shor-Laflamme enumerators, these will follow from a more general approach given in \S{\ref{section:macwilliams}}}.

\section{Scalar enumerators}\label{section:scalar}

Fix an error basis on a Hilbert space $\mathfrak{H}$. The scalars in our theory are weight enumerator polynomials with respect to this basis. We will focus on the Shor-Laflamme enumerator polynomials $A,B \in \mathbb{R}[w,z]$, but most of our results also apply the double and complete enumerator polynomials of Hu, Yang, and Yau, $C,D \in \mathbb{R}[w,x,y,z]$ and $E,F\in \mathbb{R}[\{u_R\}]$.

Write $|E|$ for the support of $E$: the set of indices $j$ with $E_j \not= I$. In particular, the weight of $E$ is the cardinality of $|E|$. Let $U = U_1 \otimes \cdots \otimes U_n$ be a local unitary transformation on $\mathfrak{H}^{\otimes n}$. Note $|U^\dagger E U| = |E|$ for any $E \in \mathcal{E}^n$. As the elements of $\mathcal{E}^n$ form an orthogonal basis (under trace-norm) for $\mathfrak{H}^{\otimes n}$ we can write 
$$U^\dagger E U = \frac{1}{q^n} \sum_{|F| = |E|} \Tr(U^\dagger E U F^\dagger) F$$
as the trace vanishes unless $|F| = |U^\dagger E U| = |E|$. Identically 
$$U^\dagger E^\dagger U = \frac{1}{q^n} \sum_{|F| = |E|} \Tr(U^\dagger E^\dagger UF) F^\dagger.$$

\begin{lemma}\label{lemma:bilinear-sum}
We have 
\begin{align*}
    &\frac{1}{q^n} \sum_{E\in \mathcal{E}^n[d]} \Tr(U^\dagger E^\dagger U F) \Tr(U^\dagger E U G^\dagger)\\
    &\qquad = \left\{\begin{array}{cl} q^n & \text{ if $F = G \in \mathcal{E}^n[d]$}\\ 0 & \text{ otherwise.}\end{array}\right.
\end{align*}
\end{lemma}
\begin{proof}
From above, for any $E \in \mathcal{E}^n[d]$ we have
\begin{align*}
    q^n &= \Tr(I) = \Tr(E^\dagger E)\\
    &= \frac{1}{q^{2n}} \sum_{F,G \in \mathcal{E}^n[d]}  \Tr(U^\dagger E^\dagger U F) \Tr(U^\dagger E U G^\dagger) \Tr(F^\dagger G)\\
    &= \frac{1}{q^n} \sum_{F \in \mathcal{E}^n[d]} \Tr(U^\dagger E^\dagger U F) \Tr(U^\dagger E U F^\dagger).
\end{align*}
Exchanging the role of $E$ and $F$ in this final expression proves the claimed result with $F = G \in \mathcal{E}^n[d]$. If $F=G$ but it is not in $\mathcal{E}^n[d]$ then we have already seen $\Tr(U^\dagger E U F^\dagger) = 0$ for any $E\in\mathcal{E}^n[d]$.

Finally, if $F\not= G$ then similar to the above
\begin{align*}
    0 &= \Tr(F G^\dagger) = \Tr(U F U^\dagger U G^\dagger U^\dagger)\\
    &= \frac{1}{q^{2n}} \sum_{E_1,E_2 \in \mathcal{E}^n[d]}  \Tr(U F U^\dagger E_1^\dagger) \Tr(U G^\dagger U^\dagger E_2) \Tr(E_1 E_2^\dagger)\\
    &= \frac{1}{q^n} \sum_{E \in \mathcal{E}^n[d]} \Tr(U F U^\dagger E^\dagger) \Tr(U G^\dagger U^\dagger E)\\
    &= \frac{1}{q^n} \sum_{E \in \mathcal{E}^n[d]} \Tr(U^\dagger E^\dagger U F) \Tr(U^\dagger E U G^\dagger)
\end{align*}
as claimed.
\end{proof}

\begin{theorem}\label{thm:scalar-unitary-invariance}
    Let $U$ be any local unitary on $\mathfrak{H}^{\otimes n}$. Then $$A(z; UM_1U^\dagger, UM_2U^\dagger) = A(z; M_1, M_2).$$
\end{theorem}
\begin{proof}
    We expand $M_1 = \frac{1}{q^n} \sum_F \Tr(F^\dagger M_1) F$ and similarly $M_2 = \frac{1}{q^n} \sum_G \Tr(G M_2) G^\dagger$. Then from the lemma,
    \begin{align*}
        &A(z; UM_1U^\dagger, UM_2U^\dagger)\\ 
        &\quad = \sum_{d=0}^n \sum_{E \in \mathcal{E}^n[d]} \Tr(U^\dagger E^\dagger U M_1) \Tr(U^\dagger E U M_2) z^d\\
        &\quad = \tfrac{1}{q^{2n}}\sum_{d=0}^n \sum_{F,G} \sum_{E \in \mathcal{E}^n[d]} \Tr(U^\dagger E^\dagger U F) \Tr(U^\dagger E U G^\dagger)\\
        &\quad\qquad\qquad\cdot \Tr(F^\dagger M_1) \Tr(G M_2) z^d\\
        &\quad = \sum_{d=0}^n \sum_{F \in \mathcal{E}^n_d} \Tr(F^\dagger M_1) \Tr(F M_2) z^d = A(z; M_1, M_2).
    \end{align*}
\end{proof}

One could provide a similar proof for the invariance of $B$, however this follows immediately from the quantum MacWilliams identity:
\begin{align*}
    &B(w,z; UM_1U^\dagger, UM_2U^\dagger)\\
    &= A\left(\tfrac{w + (\D^2-1)z}{\D}, \tfrac{w-z}{\D}; UM_1U^\dagger, UM_2U^\dagger\right)\\
    &= A\left(\tfrac{w + (\D^2-1)z}{\D}, \tfrac{w-z}{\D}; M_1,M_2\right)\\
    &= B(w,z; M_1,M_2).
\end{align*}
Note that none of $C$, $D$, $E$, or $F$ are locally unitarily invariant as they depend on counting specific Pauli operators in each factor.

\begin{lemma}\label{lemma:scalar-homomorphism}
    We have
    $$A(z; M_1\otimes M_1', M_2\otimes M_2') = A(z; M_1, M_2)\cdot A(z; M_1', M_2')$$
    and similarly for $B, C, D, E, F$.
\end{lemma}
\begin{proof}
Note that weights are additive in the sense that if $\mathrm{wt}(E_1) = d_1$ and $\mathrm{wt}(E_2) = d_2$ then $\mathrm{wt}(E_1\otimes E_2) = d_1 + d_2$. Then
\begin{align*}
    &A(z; M_1, M_2)\cdot A(z; M_1', M_2')\\
    &\ = \left(\sum_{E_1}  \Tr(E_1^\dagger M_1)\Tr(E_1M_1') z^{\mathrm{wt}(E_1)}\right)\\
    &\qquad \cdot \left(\sum_{E_2}  \Tr(E_2^\dagger M_2)\Tr(E_2M_2')) z^{\mathrm{wt}(E_2)}\right)\\
    &\ = \sum_{E_1,E_2} \Tr((E_1 \otimes E_2)^\dagger (M_1\otimes M_2))\\
    &\qquad\qquad \cdot \Tr((E_1\otimes E_2)(M_1'\otimes M_2')) z^{\mathrm{wt}(E_1) + \mathrm{wt}(E_2)}\\
    &= A(z; M_1\otimes M_1', M_2\otimes M_2').
\end{align*}
The proofs for $B, C, D, E, F$ are similar and so omitted.
\end{proof}

\begin{example}[\cite{shor1997quantum}] Let $\Pi_{\mathcal{C}}$ be the projection onto $[[n,k]]$ binary stabilizer code $\mathfrak{C}$, and write $\mathcal{S} = \mathcal{S}(\mathcal{C}) = \langle S_1, \dots, S_{n-k}\rangle$ for its stabilizer group. Then
$$\Pi_{\mathcal{C}} = \frac{1}{2^{n-k}} \prod_{j=1}^{n-k} (I + S_j) = \frac{1}{2^{n-k}} \sum_{S \in \mathcal{S}} S.$$
For any Pauli operator $P \in \mathcal{P}^n$ we have
$$\sum_{S \in \mathcal{S}} \Tr(PS) = \left\{\begin{array}{cl} 2^n & \text{ if $P \in \mathcal{S}$,}\\ 0 & \text{otherwise.}\end{array}\right.$$
Therefore
$$A_d = \frac{1}{4^{n-k}}\sum_{P \in \mathcal{P}^n[d]} \sum_{S,S'\in \mathcal{S}} \Tr(PS)\Tr(PS')
= 4^k\cdot \left| \mathcal{E}_d \cap \mathcal{S} \right|.$$
That is, properly normalized, the weight enumerator
$$\frac{1}{4^k}\cdot A(z;\mathfrak{C}) = \sum_{d=0}^n \left| \mathcal{P}^n[d] \cap \mathcal{S} \right| z^d$$
counts the number of stabilizers of each weight.

If a Pauli operator $P \in \mathcal{N} = \mathcal{N}(\mathfrak{C})$ (the normalizer of the code) then $P\Pi_\mathfrak{C} = \Pi_\mathfrak{C}P$ and so 
$$\Tr(P\Pi_\mathfrak{C}P\Pi_\mathfrak{C}) = \Tr(\Pi_\mathfrak{C}) = 2^k.$$
On the other hand, if $P \not\in \mathcal{N}$ then there is some $j_0$ where $P S_{j_0} = - S_{j_0} P$. Thus
$$P\Pi_\mathfrak{C}P\Pi_\mathfrak{C} \propto (I - S_{j_0})(I + S_{j_0})\prod_{j\not= j_0}(I \pm S_j)(I + S_j) = 0.$$
Therefore
$$B_d = \sum_{P\in \mathcal{P}^n[d]} \Tr(P\Pi_\mathfrak{C}P\Pi_\mathfrak{C}) = 2^k\cdot \left| \mathcal{E}_d \cap \mathcal{N}\right|,$$
and
$$\frac{1}{2^k} B(z;\mathfrak{C}) = \sum_{d=0}^n \left| \mathcal{P}^n[d] \cap \mathcal{N} \right| z^d$$
counts the number of normalizers of each weight. Hence, the distance of the code $\mathfrak{C}$ is the smallest $d$ such that $A_d \not= B_d$.
\end{example}

Note that the above example carries over to $\D$-ary stabilizer codes \cite{ashikhmin2001nonbinary}, adjusting the normalization appropriately. The double and complete weight enumerators further refine this information allowing separate analysis of $X$ and $Z$ errors \cite{ioffe2007asymmetric, sarvepalli2008asymmetric, stephens2008asymmetric}.

\begin{lemma}
    Let $\Pi = \ket\psi\bra\psi$ be the projector onto any state. Then $A(z;\Pi) = B(z;\Pi)$.
\end{lemma}
\begin{proof} For each $d=0,\dots, n$
    \begin{align*}
        B_d &= \sum_{E \in \mathcal{E}_d} \Tr(E^\dagger\ket\psi\bra\psi E \ket\psi\bra\psi) = \sum_{E \in \mathcal{E}_d} | \bra\psi E \ket\psi |^2  \\
        &= \sum_{E \in \mathcal{E}_d} \Tr(E^\dagger \ket\psi\bra\psi)\Tr(E\ket\psi\bra\psi) = A_d.
    \end{align*}
\end{proof}

\begin{definition}
    Let $\mathfrak{C}$ be a $[[n,k]]$ quantum code. Its \emph{encoding tensor} is the (unnormalized) Choi state 
$$\ket{T_\mathfrak{C}} = \sum_{x\in\mathbb{F}_2^k} \ket{x_L}\otimes \ket{x} = \sum_{x\in\mathbb{F}_2^k} U(\ket{x}\otimes \ket{0}^{\otimes(n-k)})\otimes \ket{x},$$
where $U$ is the unitary encoding map of the code.
\end{definition}

\begin{theorem}\label{thm:scalar-enumerator-encoded}
Let $\mathfrak{C}$ be a $[[n,1]]$-stabilizer code, $\ket{T_\mathfrak{C}}$ its encoding tensor, and $\Pi = \ket{T_\mathfrak{C}}\bra{T_\mathfrak{C}}$. Then
$$A_d(\Pi) = \tfrac{1}{4}A_d(\mathfrak{C}) + \tfrac{1}{2}B_{d-1}(\mathfrak{C}) - \tfrac{1}{4}A_{d-1}(\mathfrak{C}).$$
In particular, the distance of the code is the smallest $d$ such that $A_{d+1}(\Pi) > \tfrac{1}{4}A_{d+1}(\mathfrak{C})$.
\end{theorem}
\begin{proof}
First note that the normalizer of the code is the disjoint union
$$\mathcal{N}(\mathfrak{C}) = \mathcal{S}(\mathfrak{C}) \cup \bar{X}\mathcal{S}(\mathfrak{C}) \cup \bar{Y}\mathcal{S}(\mathfrak{C}) \cup \bar{Z}\mathcal{S}(\mathfrak{C})$$
where $\bar{X},\bar{Y},\bar{Z}$ are any choice of representations for the logical Pauli operators. Recall that $\frac{1}{2}B_d(\mathfrak{C}) = |\mathcal{N}(\mathfrak{C}) \cap \mathcal{E}_d|$ and $\frac{1}{4}A_d(\mathfrak{C}) = |\mathcal{S}(\mathfrak{C}) \cap \mathcal{E}_d|$, and so the number of proper logical operators of weight $d$ is $\tfrac{1}{2}B_{d}(\mathfrak{C}) - \tfrac{1}{4}A_{d}(\mathfrak{C})$. 

The encoding tensor state $\ket{T_\mathfrak{C}}$ is a stabilizer state (namely a $[[n+1,0]]$-stabilizer code). If we write the unencoded qubit as the $(n+1)^\text{st}$ qubit of the state, the stabilizers of this state form the disjoint union
\begin{align*}
    \mathcal{S}(\ket{T_\mathfrak{C}}) &= (\mathcal{S}(\mathfrak{C})\otimes I)  \cup (\bar{X}\mathcal{S}(\mathfrak{C}) \otimes X)\\
    & \qquad \cup (-\bar{Y}\mathcal{S}(\mathfrak{C}) \otimes Y) \cup (\bar{Z}\mathcal{S}(\mathfrak{C}) \otimes Z)
\end{align*}
where we have abused notation slightly to write, say, $\bar{X}\mathcal{S}(\mathfrak{C}) \otimes X = \{ \bar{X}S \otimes X \::\: S \in \mathcal{S}(\mathfrak{C})\}$. The weight $d$ elements of $\mathcal{S}(\mathfrak{C})\otimes I$ are naturally bijective with those of $\mathcal{S}(\mathfrak{C})$, of which there are $\frac{1}{4}A_d(\mathfrak{C})$ of them. The weight $d$ elements of $\bar{X}\mathcal{S}(\mathfrak{C}) \otimes X$ are naturally bijective with the weight $d-1$ elements of $\bar{X}\mathcal{S}(\mathfrak{C})$, and similarly for $\bar{Y}$ and $\bar{Z}$. Hence there are $\tfrac{1}{2}B_{d-1}(\mathfrak{C}) - \tfrac{1}{4}A_{d-1}(\mathfrak{C})$ of these in total.
\end{proof}

To illustrate this theorem, consider the perfect $[[5,1,3]]$ code and its encoding tensor, which is the perfect $[[6,0,4]]$ tensor, that play a prominent role in the constructions of \cite{HaPPY},
\begin{align*}
\tfrac{1}{4} \cdot A_{[[5,1,3]]}(z) &= 1 + 15 z^4,\\
\tfrac{1}{2} \cdot B_{[[5,1,3]]}(z) &= 1 + 30z^3 + 15z^4 + 18z^5\\
A_{[[6,0,4]]}(z) & = 1 + 45 z^4 + 18 z^6.
\end{align*}

Clearly this theorem cannot not hold for $k > 1$. In fact, one would not expect any natural extension of this result to $k > 1$ using the scalar enumerators we have discussed above. Consider for instance $k=2$; the encoding tensor will be stabilized by sets of operators $(\bar{X}\otimes\bar{I})\mathcal{S}(\mathfrak{C})\otimes X\otimes I$ and $(\bar{X}\otimes\bar{Y})\mathcal{S}(\mathfrak{C})\otimes X\otimes Y$, among others. In weight $d$, operators in the first set are counted in $B_d(\mathcal{C})$ and $A_{d+1}(\Pi)$ while those in the second set are also counted in $B_d(\mathcal{C})$ but in $A_{d+2}(\Pi)$. Hence to obtain such a result as in the theorem, we would need to know how to decompose $B_d(\mathcal{C})$ into counts of logical operators based on the weight of the logical operator they represent. To illustrate this point, consider the $2\times 2$ Bacon-Shor codes as a $[[4,2,2]]$ stabilizer code, and it encoding tensor as a $[[6,0,3]]$ stabilizer state, that play a prominent role in the constructions of \cite{cao2022quantum},
\begin{align*}
\tfrac{1}{16}A_{[[4,2,2]]}(z) &= 1 + 3 z^4,\\
\tfrac{1}{4}B_{[[4,2,2]]}(z) &= 1 + 18z^2 + 24z^3 + 21z^4,\\
A_{[[6,0,3]]}(z) &= 1 + 8z^3 + 21z^4 + 24z^5 + 10z^6.
\end{align*}

Nonetheless, as $A(w,z;\ket\psi\bra\psi) = B(w,z;\ket\psi\bra\psi)$ we have the $A$-enumerator of an encoding tensor (or rank-one projection in general) is invariant under the quantum MacWilliams transform, akin to the situation of classical self-dual codes. Thus enumerator polynomials of encoding tensors cannot be arbitrary, but must lie in the invariant ring of the group
$$G = \left\langle\frac{1}{q}\left(\begin{array}{cc} 1 & q^2-1 \\ 1 & -1\end{array}\right)\right\rangle.$$
As $|G| = 2$ this ring is trivial to compute and we conclude
$$A(w,z;\Pi) \in \mathbb{R}[w + (q-1)z, w^2 + (q-2)wz + (q^2 - q + 1)z^2].$$
In the case of $\D = 2$, we can make this more explicit as
\begin{align*}
    A(w,z;\ket{0}\bra{0}) = w + z,\\
    A(w,z;\ket{\beta}\bra{\beta}) = w^2 + 3z^2,
\end{align*}
where $\beta = \frac{1}{\sqrt{2}}(\ket{00} + \ket{11})$ is the Bell state. That is, the enumerator polynomial of any state lies in the ring freely generated by the enumerators of a one-qubit state and a maximally entangled two-qubit state.

Still working with $\D = 2$, we refine this somewhat for codes whose encoding tensors are ``even'' in the sense that all their stabilizers have even Hamming weight, such as the Bell state or the $[[6,0,4]]$ tensor as shown above. Then, the enumerator lies in the invariant ring of 
$$\left\langle\frac{1}{2}\left(\begin{array}{cc} 1 & 3 \\ 1 & -1\end{array}\right), \left(\begin{array}{cc} 1 & 0 \\ 0 & -1\end{array}\right)\right\rangle,$$
which is isomorphic to the dihedral group $D_6$. Again this has small order, so the invariant ring is easily computed using the algebraic software \textsc{Sage} \cite{sagemath}, yielding for any projector onto any even encoding tensor
\begin{align*}
    A(w,z;\Pi) &\in \mathbb{R}[w^2 + 3z^2, w^6 +15w^4z^2 + 15w^2z^4 + 33z^6]\\
    & \quad = \mathbb{R}[A(w,z; \ket{\beta}\bra{\beta}),\ A(w,z;\ket{T_\mathfrak{C}}\bra{T_\mathfrak{C}})]
\end{align*}
where $\ket{T_\mathfrak{C}}$ is the encoding tensor of the code $\mathfrak{C}$ with stabilizer group $\mathcal{S}(\mathfrak{C}) = \langle YYIII, YIYII, YIIYI, YIIIY \rangle$.

\section{Vector enumerators}\label{section:vector}

Ultimately, we will be interested in ``tensor'' enumerators, which are tensors whose coefficients are scalar enumerators as above. We find it illuminating to first start with the case of vector enumerators, which are single index tensor enumerators. In our first two results, we will see these do not behave like vectors, in that local unitary action transforms the vector via an adjoint representation, and that vector enumerators support a natural trace operation to scalar enumerators. Hence a vector enumerator will be presented as the coefficients of a matrix, whose rows and columns are indexed by the elements of our unitary basis $\mathcal{E}$.

Formally, let $\mathcal{E}$ be an error basis on $\mathfrak{H}$, and as before write $\mathcal{E}^{n-1} = \{ E_1 \otimes \dots \otimes E_{n-1} \::\: E_1, \dots, E_{n-1} \in \mathcal{E}\}$ and $\mathcal{E}^{n-1}[d] = \{E \in \mathcal{E}^{n-1}\::\: \mathrm{wt}(E) = d\}$. Fix a ``leg'' $j \in \{1,\dots, n\}$ and for $E \in \mathcal{E}$ and $F \in \mathcal{E}^{n-1}$ define
$$E \otimes_j F = F_1 \otimes \cdots \otimes F_{j-1} \otimes E \otimes F_j \otimes \cdots F_{n-1} \in \mathcal{E}^n.$$
That is, $E \otimes_j F$ has $E$ is inserted as the $j$-th factor into $F$. For a pair of elements of this basis $(E,E')$ define weights
\begin{align}
    &A^{(j)}_d(E,E';M_1,M_2) \label{equation:vector-A-weight}\\
    &\quad = \sum_{F \in \mathcal{E}^{n-1}[d]} \Tr((E \otimes_j F)^\dagger M_1) \Tr((E' \otimes_j F) M_2),\nonumber\\
    &B^{(j)}_d(E,E';M_1,M_2)\label{equation:vector-B-weight}\\
    &\quad = \sum_{F \in \mathcal{E}^{n-1}[d]} \Tr((E \otimes_j F)^\dagger M_1 (E' \otimes_j F) M_2).\nonumber
\end{align}
These enumerators are related to the usual $A$-weight (\ref{equation:scalar-A-weight}) and $B$-weight (\ref{equation:scalar-B-weight}) of $M_1$ and $M_2$, refined by specifying the error operators $E,E'$ that appear in factor $j$. Note that the weights of the operators $E$ and $E'$ do not contribute to $A^{(j)}_d$ or $B^{(j)}_d$.

These enumerators are indexed by $(E,E')$ and so may be viewed as the coefficients of a matrix. Let $e_{E,E'}$ be the matrix unit, which is $1$ in $(E,E')$ position and zero elsewhere. Write $V = \mathrm{span}_\mathbb{R}\{e_{E,E'} \::\: E,E' \in \mathcal{E}\}$ for this matrix space.

\begin{definition}
    The \emph{vector enumerators} of ``leg'' $j$ are
    \begin{align*}
        \mathbf{A}^{(j)}(z; M_1, M_2) &= \sum_{E,E'\in\mathcal{E}} \sum_{d=0}^n A^{(j)}_d(E,E';M_1,M_2) z^d e_{E,E'},\\
        \mathbf{B}^{(j)}(z; M_1, M_2) &= \sum_{E,E'\in\mathcal{E}} \sum_{d=0}^n B^{(j)}_d(E,E';M_1,M_2) z^d e_{E,E'},
    \end{align*}
    where $A^{(j)}_d$ and $B^{(j)}_d$ are defined in (\ref{equation:vector-A-weight}) and (\ref{equation:vector-B-weight}).
\end{definition}

These enumerators are elements of $\mathbb{R}[z]\otimes_\mathbb{R} V$, which is to say vectors in $V$ whose coefficients are real polynomials in $z$. We can trace vector enumerators down to scalar enumerators, which is captured in the following result, whose proof is obvious.

\begin{lemma}\label{lemma:vector-to-scalar}
    Consider $\widetilde{\mathrm{tr}}: \mathbb{R}[z]\otimes_\mathbb{R} V \to \mathbb{R}[z]$ by linearly extending
    $$\widetilde{\mathrm{tr}}(e_{E,E'}) = \left\{\begin{array}{cl} 1 & \text{ if $E = E' = I$,}\\ z & \text{ if $E = E' \not= I$,} \\ 0 & \text{ if $E \not= E'$.}\end{array}\right.$$
    Then $\widetilde{\mathrm{tr}}(\mathbf{A}^{(j)}(z)) = A(z)$ and $\widetilde{\mathrm{tr}}(\mathbf{B}^{(j)}(z)) = B(z)$.
\end{lemma}

Note that vector enumerators are not invariant under local unitary operations. Rather, they are covariant with respect to an adjoint, or ``Bloch sphere'' style, representation of unitaries on $V$. Specifically, let $U$ be a unitary on $\mathfrak{H}$ and define $c_{EF}\in \mathbb{C}$ by $U^\dagger E U = \sum_{G} G c_{GE}$, where $c_{GE} = \frac{1}{\D}\Tr(U^\dagger E U G^\dagger)$. We also have $U^\dagger E^\dagger U = \sum_{G} G^\dagger c^*_{GE}$, and $UEU^\dagger = \sum_G c^*_{EG} G$ and $UE^\dagger U^\dagger = \sum_G c_{EG} G^\dagger$. Now define the action of $U$ on the matrix space $V$ by
$$\Lambda(U)(e_{E,E'}) = \sum_{G,G'} c^*_{EG} c_{E'G'} e_{G,G'}.$$
We will formalize this in the next section. In particular, the following result is a special case of Theorem \ref{theorem:tensor-A-unitary-covariance} below and so we do not provide a proof here.

\begin{theorem}\label{theorem:vector-LU-covariant}
    Let $U = U_1\otimes \cdots \otimes U_n$ be a local unitary operator. Then $$\mathbf{A}^{(j)}(z;UM_1U^\dagger , UM_2U^\dagger) = \Lambda(U_j)\left(\mathbf{A}^{(j)}(z;M_1,M_2)\right).$$
\end{theorem}

Just as scalar enumerators, vector enumerators behave nicely under tensor product. We state this formally as follows. Again, the proof of this result follows immediately from the general case of tensor enumerators, Proposition \ref{proposition:tensor-homomorphism}, and so is omitted.

\begin{proposition}\label{prop:vector-functoral}
    Let $M_1\otimes N_1$ and $M_2\otimes N_2$ be Hermitian operators on $\mathfrak{H} \otimes \mathfrak{K} $ where $\mathfrak{H} = \bigotimes_{j=1}^n \mathfrak{H}_j$ and $\mathfrak{K} = \bigotimes_{k=1}^m \mathfrak{K}_k$. Then for $j = 1, \dots, n$ we have
    \begin{align*}
        &\mathbf{A}^{(j)}(z;M_1\otimes N_1,M_2\otimes N_2)\\
        &\qquad = \mathbf{A}^{(j)}(z;M_1,M_2)A(z;N_1,N_2).
    \end{align*}
    and for $j=n+1, \dots, n+m$ we have
    \begin{align*}
        &\mathbf{A}^{(j)}(z;M_1\otimes N_1,M_2\otimes N_2)\\
        &\qquad = A(z;M_1,M_2)\mathbf{A}^{(j-n)}(z;N_1,N_2).
    \end{align*}
    Analogous formulas hold for $\mathbf{B}^{(j)}$ and $B$.
\end{proposition}

When $M_1 = M_2 = \Pi_\mathfrak{C}$ the projection onto a binary stabilizer code $\mathfrak{C}$, the vector enumerators carry refined information about the structure of the stabilizer and normalizer groups of $\mathfrak{C}$. In general, vector enumerators are matrices, however as long as the underlying code has distance $d \geq 2$ these matrices are diagonal except in a degenerate case (see Theorem \ref{theorem:vector-diagonal} below). The coefficients along the diagonal are polynomials that are themselves are a form of scalar enumerator. For $P = I, X, Y, Z$, the coefficient of $e_{P,P}$ in $\mathbf{A}^{(j)}$ (respectively $\mathbf{B}^{(j)}$) counts the number of stabilizers (respectively logical operators) of each weight with the side condition that its $j^{\text{th}}$ factor is $P$.

\begin{example}
Consider the vector enumerator of the logical leg in an encoding of a $[[n,1]]$ binary stabilizer code $\mathfrak{C}$. Recall from the proof of Theorem \ref{thm:scalar-enumerator-encoded} above that if $\ket{T_\mathfrak{C}}$ is the encoding state, then the stabilizer of this state is the disjoint union
\begin{align*}
    \mathcal{S}(\ket{T_\mathfrak{C}}) &= (\mathcal{S}(\mathfrak{C})\otimes I)  \cup (\bar{X}\mathcal{S}(\mathfrak{C}) \otimes X)\\
    & \qquad \cup (-\bar{Y}\mathcal{S}(\mathfrak{C}) \otimes Y) \cup (\bar{Z}\mathcal{S}(\mathfrak{C}) \otimes Z)
\end{align*}
and the projection onto this state is
\begin{align*}
    \ket{T_\mathfrak{C}}\bra{T_\mathfrak{C}} &= \frac{1}{2^{n+1}}\sum_{S\in \mathcal{S}(\mathfrak{C})} S\otimes I + \bar{X}S \otimes X\\
    &\qquad\qquad\qquad - \bar{Y}S\otimes Y + \bar{Z}S\otimes Z.
\end{align*}
Now, for any $P = I,X,Y,Z$ we have
$$\Tr((F\otimes P)\ket{T_\mathfrak{C}}\bra{T_\mathfrak{C}}) = \left\{\begin{array}{cl} \pm 1 & \text{if $F \in \bar{P}\mathcal{S}(\mathfrak{C})$}\\ 0 & \text{otherwise.}\end{array}\right.$$
Therefore
\begin{align*}
    &\mathbf{A}^{(n+1)}(z;\ket{T_\mathfrak{C}}\bra{T_\mathfrak{C}})\\
    &\qquad = \sum_{P=I,X,Y,Z} \sum_{d=0}^n | \bar{P}\mathcal{S}(\mathfrak{C}) \cap \mathcal{E}^n_d | z^d e_{P,P}.
\end{align*}
In particular,
\begin{align*}
    A(\mathfrak{C}) &= \mathbf{A}^{(n+1)}(I,I;\ket{T_\mathfrak{C}}\bra{T_\mathfrak{C}}) \text{ and }\\
    B(\mathfrak{C}) &= \sum_{P=I,X,Y,Z} \mathbf{A}^{(n+1)}(P,P;\ket{T_\mathfrak{C}}\bra{T_\mathfrak{C}}).
\end{align*}
\end{example}

\begin{lemma}\label{lemma:diagonal-terms}
    Let $\mathfrak{C}$ be a $[[n,k]]$ stabilizer code. Then for $P\not= P'\in \mathcal{P}$,
    \begin{enumerate}
        \item $\mathbf{A}^{(j)}(P,P')$ is nonzero if and only if $PP'\otimes_j I^{\otimes(n-1)}$ is a stabilizer;
        \item $\mathbf{B}^{(j)}(P,P')$ is nonzero if and only if $\mathbf{A}^{(j)}(P,P')$ is nonzero.
    \end{enumerate}
\end{lemma}
\begin{proof}
    Recall that if $\Pi = \frac{1}{2^{n-k}} \sum_{S \in \mathcal{S}(\mathfrak{C})} S$ is the projector onto the code space, then $\Tr((P \otimes_j Q)\Pi) \not= 0$ if and only if $P \otimes_j Q \in \mathcal{S}(\mathfrak{C})$. Now $\mathbf{A}^{(j)}(P,P')$ is nonzero if and only if there exists some $Q\in\mathcal{P}^{n-1}$ such that
    $$\Tr((P \otimes_j Q)\Pi)\Tr((P' \otimes_j Q)\Pi) \not= 0,$$
    and hence both $(P \otimes_j Q), (P' \otimes_j Q) \in \mathcal{S}(\mathfrak{C})$ and so is their product $PP' \otimes_j I^{\otimes(n-1}$.

    Now $\mathbf{B}^{(j)}(P,P') \not= 0$ if and only if there is a $Q$ such that
    $$\Tr((P \otimes_j Q)\Pi(P' \otimes_j Q)\Pi) \not= 0,$$
    which happens if and only if $(P \otimes_j Q) \in \mathcal{N}(\mathfrak{C})$, and
    $$(P \otimes_j Q)(P' \otimes_j Q) = (PP' \otimes_j I^{\otimes(n-1}) \in \mathcal{S}(\mathfrak{C}).$$
\end{proof}

\begin{theorem}\label{theorem:vector-diagonal}
    Let $\mathfrak{C}$ be a $[[n,k,d]]$ stabilizer code of distance $d\geq 2$. Then every for every $j=1, \dots, n$ we have either:
    \begin{enumerate}
        \item there is a Pauli operator $P$ and code $\mathfrak{C}'$ such that
        \begin{align*}
            \mathbf{A}^{(j)}(\mathfrak{C}) &= A(\mathfrak{C}')\cdot(e_{I,I} + e_{I,P} + e_{P,I} + e_{P,P}),\\
            \mathbf{B}^{(j)}(\mathfrak{C}) &= B(\mathfrak{C}')\cdot(e_{I,I} + e_{I,P} + e_{P,I} + e_{P,P});
        \end{align*}
        \item or $\mathbf{A}^{(j)}(\mathfrak{C})$ and $\mathbf{B}^{(j)}(\mathfrak{C})$ are diagonal.
    \end{enumerate}
\end{theorem}
\begin{proof}
    From the lemma, if $\mathbf{A}^{(j)}(\mathfrak{C})$ has off-diagonal terms then $P \otimes_j I^{\otimes (n-1)} \in \mathcal{S}(\mathfrak{C})$ for some $P$. This implies $\mathfrak{C} = \ket{\psi}\otimes_j\mathfrak{C}'$ where $\ket{\psi}$ is the $+1$-eigenstate of $P$ and $\mathfrak{C}'$ is a $(n-1)$-qubit stabilizer code. By Proposition \ref{prop:vector-functoral} we have $\mathbf{A}^{(j)}(\mathfrak{C}) = \mathbf{A}^{(j)}(\ket{\psi}\bra{\psi}) A(\mathfrak{C}')$ and $\mathbf{B}^{(j)}(\mathfrak{C}) = \mathbf{B}^{(j)}(\ket{\psi}\bra{\psi}) B(\mathfrak{C}')$, and it is straightforward to show $$\mathbf{A}^{(j)}(\ket{\psi}\bra{\psi}) = \mathbf{B}^{(j)}(\ket{\psi}\bra{\psi}) = e_{I,I} + e_{I,P} + e_{P,I} + e_{P,P}.$$
\end{proof}

Recall the ``cleaning lemma'' \cite[Lemma 1]{bravyi2009no} relates the supports of logical Pauli operators on a stabilizer code. It states: for any subset of qubits $J$ either there exists a nontrivial logical Pauli operators supported on $J$, or every logical Pauli operator has a representation is that is trivial on $J$ (that is it can be ``cleaned'' from $J$). In particular, if $|J| < d$, the distance of the code, then no nontrivial logical operator can be supported on $J$. Thus any subset of qubits of size less than the distance of the code can be cleaned. We are interested in a stronger notion of cleaning: every Pauli operator on $J$ extends to a stabilizer of the code. As there are $4^{|J|}$ Pauli operators on $J$ but only $2^{n-k}$ stabilizers, we must have $|J| \leq \frac{1}{2}(n-k)$. In this case, we can clean logical operators from $J$ as follows: given any logical Pauli operator $L$, we take the Pauli operator on $J$ given by $\left. L \right|_J$, and extend this to a stabilizer $S$. Then $LS$ is equivalent to $L$ and trivial on $J$.

\begin{corollary}
    Let $\mathfrak{C}$ be a stabilizer code and $\ket{T_\mathfrak{C}}$ its encoding tensor. Then for any logical qubit $j$ we have $\mathbf{A}^{(j)}(\ket{T_\mathfrak{C}}\bra{T_\mathfrak{C}})$ is diagonal.
\end{corollary}
\begin{proof}
    By construction, along every logical leg $j$ every local Pauli operator gives rise to a stabilizer of $\ket{T_\mathfrak{C}}$, and thus logical legs can be cleaned. Moreover, these stabilizers are not of the form $P\otimes_j I^{\otimes(n-1)}$ and therefore in the theorem we are in the case of $\mathbf{A}^{(j)}(\ket{T_\mathfrak{C}}\bra{T_\mathfrak{C}})$ being diagonal.
\end{proof}

\begin{definition}
    The \emph{reduced} enumerator $\tilde{\mathbf{A}}^{(j)}$ is the diagonal of $\mathbf{A}^{(j)}$. Namely, if 
    $\textbf{A}^{(j)} = \sum_{E,E'} \sum_{d=0}^{n-1} A^{(j)}_d(E,E')z^d e_{E,E'}$
    then $$\tilde{\textbf{A}}^{(j)} = \sum_{E} \sum_{d=0}^{n-1} A^{(j)}_d(E)z^d e_{E,E}.$$
    Identically, $\tilde{\mathbf{B}}^{(j)}$ is the restriction of $\textbf{B}^{(j)}$ to the diagonal.  
\end{definition}

The content of the above theorem can be captured as: if $d \geq 2$ then, except in a trivial case, $\textbf{A}^{(j)} = \tilde{\textbf{A}}^{(j)}$ (and similarly $\textbf{B}^{(j)} = \tilde{\textbf{B}}^{(j)}$) for each $j=1, \dots, n$. To simplify the notation, we will write $e_E = e_{E,E}$ when working with reduced enumerators.

Next we present the analogue of the quantum MacWilliams identity for vector enumerators. This theorem is a special case of that for tensor enumerators and hence we omit the proof. We hasten to point out that the components of $\mathbf{A}^{(j)}$ and $\mathbf{B}^{(j)}$ are polynomials of degree at most $n-1$, hence their homogenization are of degree exactly $n-1$. That the overall degree is one less than that of the scalar enumerators is the genesis of an additional factor of $\frac{1}{q}$, which has been incorporated into the map $\Psi$.

\begin{theorem}\label{theorem:vector-macwilliams-identity}
    For an error basis $\mathcal{E}$ on $\mathfrak{H} = \mathbb{C}^\D$ and $M_1,M_2$ being Hermitian operators on $\mathfrak{H}^{\otimes n}$, each leg $j=1,\dots,n$ has
    $$\mathbf{B}^{(j)}(w,z;M_1,M_2) = \Psi\left[ \mathbf{A}^{(j)}(\tfrac{w+(q^2-1)z}{q},\tfrac{w-z}{q};M_1,M_2) \right]$$
    where $\Psi(e_{E,E'}) = \frac{1}{q^2}\sum_{F,F'} \Tr(F^\dagger E F' (E')^\dagger) e_{F,F'}$.
\end{theorem}

While this theorem focuses on vector enumerator analogues of the Shor-Laflamme enumerator polynomials, there are vector enumerator extensions of the double and complete enumerators by Hu, Yang, and Yau, $C$, $D$, $E$, and $F$. Similarly $\mathbf{C}^{(j)}$ and $\mathbf{D}^{(j)}$, and $\mathbf{E}^{(j)}$ and $\mathbf{F}^{(j)}$, are related by MacWilliams identities. The form of these identities, and their proofs, is covered in \S\ref{section:macwilliams} below.

Note that the diagonal elements $e_{E,E}$ are left invariant under $\Psi$. Namely we compute
\begin{align*}
    \Psi(e_{E,E}) &= \tfrac{1}{\D^2} \sum_{F,F'}\Tr(F^\dagger E F' E^\dagger)e_{F,F'}\\
    &= \tfrac{1}{\D^2}\sum_{F,F'}\omega(E,F')\Tr(F^\dagger F')e_{F,F'}\\
    &= \tfrac{1}{\D}\sum_F \omega(E,F)e_{F,F}.
\end{align*}
This is the discrete Wigner transform on $\mathcal{E}$, and hence we have the following result.

\begin{corollary}
    Let $\mathcal{E}$ be an error basis on $\mathfrak{H} = (\mathbb{C}^\D)$. Then for any Hermitian operators $M_1, M_2$ on $\mathfrak{H}^{\otimes n}$ and leg $j=1,\dots,\D$ we have
    $$\tilde{\mathbf{B}}^{(j)}(w,z;M_1,M_2) = \Psi\left[ \tilde{\mathbf{A}}^{(j)}(\tfrac{w+(q^2-1)z}{q},\tfrac{w-z}{q};M_1,M_2) \right]$$
    where $\Psi(e_E) = \tfrac{1}{\D}\sum_{F\in\mathcal{E}} \omega(E,F)e_{F}$ is the discrete Wigner transform on $\mathcal{E}$.
\end{corollary}

The vector analogue of the quantum MacWilliams identity above allows us prove that the vector enumerators $\mathbf{B}^{(j)}$ behave contravariantly, analogous to the covariance Theorem \ref{theorem:vector-LU-covariant} above for $\mathbf{A}^{(j)}$.

\begin{lemma}\label{lemma:Wigner-transform-covariance}
    $\Psi\circ \Lambda(U) = \Lambda(U^\dagger)\circ\Psi$ for any unitary $U$.
\end{lemma}
\begin{proof}
By direct computation:
    \begin{align*}
        &\Psi\left[\Lambda(U)(e_{E,E'}))\right] = \sum_{G,G'} c_{EG}^* c_{E'G'} \Psi(e_{G,G'})\\
        &\quad = \sum_{G,G'} \sum_{F,F'} \Tr(F^\dagger G F' (G')^\dagger) c_{EG}^*c_{E'G'} e_{F,F'}\\
        &\quad = \sum_{F,F'} \Tr\left[F^\dagger \left(\sum_G c_{EG}^* G\right) F' \left(\sum_{G'} c_{E'G'} (G')^\dagger\right)\right] e_{F,F'}\\
        &\quad = \sum_{F,F'} \Tr[F^\dagger (UEU^\dagger) F' (U(E')^\dagger U^\dagger) ] e_{F,F'}\\
        &\quad = \sum_{F,F'} \Tr[(U^\dagger F^\dagger U) E (U^\dagger F' U) (E')^\dagger) ] e_{F,F'}\\
        &\quad = \sum_{F,F'}\sum_{G,G'} \Tr[G^\dagger E G' (E')^\dagger ] c_{FG}c_{F'G'}^* e_{F,F'}\\
        &\quad = \sum_{G,G'} \Tr(G^\dagger E G' (E')^\dagger ] \Lambda(U^\dagger)(e_{G,G'})\\
        &\quad = \Lambda(U^\dagger)(\Psi(e_{G,G'})).
    \end{align*}
\end{proof}

\begin{theorem}
    Let $U = U_1\otimes \cdots \otimes U_n$ be a local unitary transformation. Then $$\mathbf{B}^{(j)}(z;UM_1U^\dagger , UM_2U^\dagger) = \Lambda(U_j^\dagger)\left(\mathbf{B}^{(j)}(z;M_1,M_2)\right).$$
\end{theorem}
\begin{proof}
    From the lemma, and Theorems \ref{theorem:vector-LU-covariant} and \ref{theorem:vector-macwilliams-identity}:
    \begin{align*}
        &\mathbf{B}^{(j)}(z;UM_1U^\dagger, UM_2U^\dagger)\\
        &\quad = \Psi\left[ \mathbf{A}^{(j)}(\tfrac{w+(q^2-1)z}{q},\tfrac{w-z}{q};UM_1U^\dagger, UM_2U^\dagger) \right]\\
        &\quad = \Psi\left[ \Lambda(U_j) \mathbf{A}^{(j)}(\tfrac{w+(q^2-1)z}{q},\tfrac{w-z}{q};M_1,M_2) \right]\\
        &\quad = \Lambda(U_j^\dagger) \Psi\left[\mathbf{A}^{(j)}(\tfrac{w+(q^2-1)z}{q},\tfrac{w-z}{q};M_1,M_2) \right]\\
        &\quad = \Lambda(U_j^\dagger) \mathbf{B}^{(j)}(z;M_1, M_2).
    \end{align*}
\end{proof}

\section{Tensor enumerators}\label{section:tensor}

Tensor enumerators are the natural extension of the vector enumerators of the previous section to multiple legs. Consider a subset $J = \{ j_1,\dots, j_m\} \subseteq \{1,\dots,n\}$. Then extending the notation above, for $E\in\mathcal{E}^m$ and $F\in\mathcal{E}^{n-m}$ write $E \otimes_J F$ for the operator where for each $k=1, \dots, m$ we insert $E_k$ as the $j_k$-th factor into $F$. Our tensor enumerators will be associated to the space $V^{\otimes m}$, where $V = \mathrm{span}_\mathbb{R}\{e_{E,E'} \::\: E,E'\in\mathcal{E}\}$; we will find it convenient to reorder the indexing of these formal basis elements so that for $E,E'\in\mathcal{E}^m$ we have
$$e_{E,E'} = e_{E_1,E'_1} \otimes \cdots \otimes e_{E_m,E'_m} \in V^{\otimes m}.$$
Tensor enumerators will be elements of $\mathbb{R}[z]\otimes_\mathbb{R}V^{\otimes m}$. That is, they are tensors whose coefficients are weight enumerators. 

Formally define the weights
\begin{align}
    &A^{(J)}_d(E,E';M_1,M_2)\label{equation:tensor-A-weight}\\
    &\quad = \sum_{F \in \mathcal{E}^{n-m}[d]} \Tr((E \otimes_J F)^\dagger M_1) \Tr((E' \otimes_J F) M_2),\nonumber\\
    &B^{(J)}_d(E,E';M_1,M_2)\label{equation:tensor-B-weight}\\
    &\quad = \sum_{F \in \mathcal{E}^{n-m}[d]} \Tr((E \otimes_J F)^\dagger M_1 (E' \otimes_J F) M_2).\nonumber
\end{align}
As with the vector enumerators, weights of $E$ and $E'$ are not included in these expressions.

\begin{definition}
    The \emph{tensor enumerators} of set of legs $J$ are then
    \begin{align*}
        &\mathbf{A}^{(J)}(z; M_1, M_2) \\
        &\qquad = \sum_{E,E'\in\mathcal{E}^m} \sum_{d=0}^{n-m} A^{(J)}_d(E,E';M_1,M_2) z^d e_{E,E'},\\
        &\mathbf{B}^{(J)}(z; M_1, M_2) \\
        &\qquad = \sum_{E,E'\in\mathcal{E}^m} \sum_{d=0}^{n-m} B^{(J)}_d(E,E';M_1,M_2) z^d e_{E,E'},
    \end{align*}
    where $A^{(J)}_d$ and $B^{(J)}_d$ are defined in (\ref{equation:tensor-A-weight}) and (\ref{equation:tensor-B-weight}).
\end{definition}

We can extend the weighted trace of Lemma \ref{lemma:vector-to-scalar} above to link tensor enumerators of different degrees. Let $K \subset J$, with $|K| = k$ and $|J| = m$. Then any element of $\mathcal{E}^m$ can be written as $F \otimes_K E$ for $F\in\mathcal{E}^k$ and $E \in \mathcal{E}^{m-k}$. We define $\widetilde{\mathrm{tr}}^J_K: \mathbb{R}[z]\otimes_\mathbb{R} V^{\otimes m} \to \mathbb{R}[z]\otimes_\mathbb{R} V^{\otimes k}$ by linearly extending
$$\widetilde{\mathrm{tr}}^J_K(e_{F \otimes_K E,F' \otimes_K E'}) = \left\{\begin{array}{cl} z^{\mathrm{wt}(E)}e_{F,F'} & \text{if $E = E'$}\\ 0 & \text{if $E \not= E'$.}\end{array}\right.$$
The proof of the following result is a straightforward computation and so is left for the reader.

\begin{proposition}
    Let $K \subset J$. Then $\widetilde{\mathrm{tr}}^J_K(\mathbf{A}^{(J)}(z)) = \mathbf{A}^{(K)}(z)$, and similarly for $\mathbf{B}^{(J)}$.
\end{proposition}

As in the previous section, the action of the unitary group on our error basis naturally defines a representation on the vector space $V$. This extends naturally to local unitary operators on $V^{\otimes m}$. In fact, the following definition makes sense and the theorem holds for any $U\in \mathcal{U}(\mathfrak{H}^{\otimes m})$, not just local ones.

\begin{definition}\label{definition:unitary-representation}
    Let $\mathcal{E}$ be an error group on a Hilbert space $\mathcal{H}$, and $V = \mathrm{span}\{e_{E,E'}\}_{E,E'\in\mathcal{E}}$. Define $\Lambda: \mathcal{U}(\mathfrak{H}) \to L(V)$ by
    \begin{equation*}
        \Lambda(U)(e_{E,E'}) = \sum_{G,G'\in\mathcal{E}} c_{EG}^* c_{E'G'} e_{G,G'},
    \end{equation*}
    where $U^\dagger E U = \sum_G G c_{GE}$. Then (without additional decoration) $\Lambda:\mathcal{U}(\mathfrak{H})^{\otimes m} \to L(V^{\otimes m})$ as the $m$-fold tensor product of $\Lambda$.
\end{definition}

Note that in this definition, $c_{GE} = \frac{1}{\D}\Tr(G^\dagger U^\dagger E U)$. By exchanging the roles of $E$ and $G$, and conjugating, we get $c^*_{EG} = \frac{1}{\D}\Tr(G^\dagger U E U^\dagger)$ and thus
\begin{equation*}
    \Lambda(U^\dagger)(e_{E,E'}) = \sum_{G,G'\in\mathcal{E}} e_{G,G'} c_{GE} c^*_{G'E'}.
\end{equation*}

\begin{theorem}\label{theorem:tensor-A-unitary-covariance}
    Let $U = U_1 \otimes \cdots \otimes U_n$ be a local unitary operator on $\mathfrak{H}^{\otimes n}$, and let $J \subseteq \{1, \dots, n\}$ be any subset of legs. Then
    $$\mathbf{A}^{(J)}(z; UM_1 U^\dagger, UM_2 U^\dagger) = \Lambda(U_J)\left(\mathbf{A}^{(J)}(z; M_1,M_2)\right),$$
    where $U_J = \bigotimes_{j\in J} U_j$.
\end{theorem}
\begin{proof}
We write
\begin{align*}
    &\mathbf{A}^{(J)}(z;U M_1 U^\dagger, U M_2 U^\dagger)\\
    &\quad = \sum_{E,E'} \sum_{F} \Tr((E^\dagger \otimes_J F^\dagger) U M_1 U^\dagger )\\
    &\qquad\qquad \cdot \Tr((E' \otimes_J F) U M_2 U^\dagger) z^{\mathrm{wt}(F)} e_{E,E'}\\
    &\quad = \ \sum_{E,E'} \sum_{F} \Tr((U_J^\dagger E U_J \otimes_J \hat{U} F \hat{U})^\dagger M_1)\\
    &\qquad\qquad \cdot \Tr((U_J^\dagger E' U_J \otimes_J \hat{U}^\dagger F \hat{U}) M_2) z^{\mathrm{wt}(F)} e_{E,E'}
\end{align*}
where $U_J = \bigotimes_{j\in J} U_j$ and $\hat{U} = \bigotimes_{j \not\in J} U_j$. The invariance of this expression under $\hat{U}$ follows exactly as in Theorem \ref{thm:scalar-unitary-invariance}. For $U_J$ we expand as above
\begin{align*}
    &\mathbf{A}^{(J)}(z;U M_1 U^\dagger, U M_2 U^\dagger)\\
    &= \sum_{E,E',G,G'} \sum_{F} c^*_{GE} c_{G'E'} \Tr((G \otimes_J F)^\dagger M_1)\\
    &\qquad \qquad \cdot \Tr((G' \otimes_J F) M_2) z^{\mathrm{wt}(F)} e_{E,E'}\\
    &= \sum_{G,G'} \sum_{F}\Tr((G \otimes_J F)^\dagger M_1) \Tr((G' \otimes_J F) M_2) \\
    &\qquad \qquad \cdot z^{\mathrm{wt}(F)} \Lambda(U_J)(e_{G,G'})\\
    &= \Lambda(U_J)\left(\mathbf{A}^{(J)}(z;M_1,M_2)\right).
\end{align*}
\end{proof}

The weighted trace operation is compatible with local unitary action, as the following result shows.

\begin{proposition}
    Let $U = U_1 \otimes \cdots \otimes U_n$ be a local unitary operator on $\mathfrak{H}^{\otimes n}$, and let $K \subset J \subseteq \{1, \dots, n\}$. Then $$\Lambda(U_K) \circ \widetilde{\mathrm{tr}}^J_K = \widetilde{\mathrm{tr}}^J_K\circ \Lambda(U_J).$$
\end{proposition}
\begin{proof}
    Write $k = |K|$ and $m = |J|$. Then the left side of this equation on a basis element of $V^{\otimes m}$ is 
    \begin{align*}
        &\Lambda(U_K)[\widetilde{\mathrm{tr}}^J_K(e_{E\otimes_K F, E'\otimes_K F})] = z^{\mathrm{wt}(F)} \Lambda(U_K)(e_{E,E'})\\
        &\quad = z^{\mathrm{wt}(F)} \sum_{G,G'\in\mathcal{E}^k} c^*_{EG} c_{E'G'} e_{G,G'},
    \end{align*}
    while $\Lambda(U_K)[\widetilde{\mathrm{tr}}^J_K(e_{E\otimes_K F, E'\otimes_K F'})] = 0$ when $F\not= F'$.

    Now we compute
    \begin{align*}
        &\widetilde{\mathrm{tr}}^J_K[\Lambda(U_J)(e_{E\otimes_K F, E\otimes_K F'})]\\
        &\quad = \sum_{G,G',H,H'} c^*_{EG}c^*_{FH}c_{E'G'}c_{F'H'} \widetilde{\mathrm{tr}}^J_K(e_{G\otimes_K H, G'\otimes H'}),\\
        &\quad = \sum_{H\in\mathcal{E}^{m-k}} z^{\mathrm{wt}(H)} c^*_{FH}c_{F'H}\cdot \sum_{G,G'\in\mathcal{E}^k} c^*_{EG}c_{E'G'} e_{G,G'}.
    \end{align*}
    From Lemma \ref{lemma:bilinear-sum}, 
    \begin{align*}
        &\sum_{H\in\mathcal{E}^{m-k}[d]} \Tr(U^\dagger H^\dagger U F) \Tr(U^\dagger H U (F')^\dagger)\\
        &\quad = q^{2(m-k)}\sum_{H\in\mathcal{E}^{m-k}[d]} c^*_{FH} c_{F'H}\\
        &\quad = \left\{ \begin{array}{cl} q^{2(m-k)} & \text{if $F = F'\in\mathcal{E}^{m-k}[d]$} \\ 0 & \text{otherwise.}\end{array}\right.
    \end{align*}
    Therefore
    $$\sum_{H\in\mathcal{E}^{m-k}} z^{\mathrm{wt}(H)} c^*_{FH}c_{F'H} = \left\{ \begin{array}{cl} z^{\mathrm{wt}(F)} & \text{if $F = F'$} \\ 0 & \text{otherwise,}\end{array}\right.$$
    and so the right side of the equation coincides with the left on every element of $V^{\otimes m}$.
\end{proof}

A crucial property we will need is that tensor enumerators are homomorphic under the tensor product. That is, the tensor enumerator of a tensor product is the tensor product of the individual enumerators. Formally this captured in the following result.

\begin{proposition}\label{proposition:tensor-homomorphism}
    Let $M_1,M_2$ and $N_1,N_2$ be Hermitian operators on Hilbert spaces $\mathfrak{H}$ and $\mathfrak{K}$ respectively, and let $J$ and $K$ be subsets of legs on each of these spaces. Then
    \begin{align*}
       & \mathbf{A}^{(J\cup K)}(z;M_1\otimes N_1, M_2\otimes N_2)\\
       &\qquad =\  \mathbf{A}^{(J)}(z;M_1, M_2) \otimes \mathbf{A}^{(K)}(z;N_1,N_2),
    \end{align*}
    and similarly for $\mathbf{B}$.
\end{proposition}
\begin{proof}
    Just as in the proof of Lemma \ref{lemma:scalar-homomorphism} this is a direct computation:
    \begin{align*}
        &\mathbf{A}^{(J\cup K)}(z;M_1\otimes N_1, M_2\otimes N_2)\\
        &= \sum_{\scriptsize\begin{array}{c@{,\:}c} E & E'\\ F & F'\end{array}} \sum_{G,H} \Tr\left[((E \otimes_J G) \otimes (F\otimes_{K} H))^\dagger (M_1\otimes N_1)\right]\\
        &\qquad\qquad \cdot \Tr\left[((E' \otimes_J G) \otimes (F' \otimes_{K} H))(M_2\otimes N_2)\right]\\
        &\qquad\qquad \cdot z^{\mathrm{wt}(G)+\mathrm{wt}(H)} e_{E,E',F,F'}\\
        & \quad = \sum_{E,E',G} \Tr((E\otimes_J G)^\dagger M_1) \Tr((E'\otimes_J G) M_2) z^{\mathrm{wt}(G)}\\
        &\qquad \cdot \sum_{F,F',H} \Tr((F\otimes_{K} H)^\dagger N_1) \Tr((F' \otimes_{K} H) M'_2) z^{\mathrm{wt}(H)}\\
        &\qquad\qquad \cdot e_{E,E'}\otimes e_{F,F'}\\
        &\quad = \mathbf{A}^{(J)}(z; M_1, M_2) \otimes \mathbf{A}^{(K)}(z; N_1, N_2).
    \end{align*}
\end{proof}

In Lemma \ref{lemma:diagonal-terms} above, we showed that a nonzero off-diagonal term in a vector enumerator indicates the existence of a stabilizer supported on the leg of the vector enumerator. This results carries over to general tensor enumerators with an identical proof, which we state below for completeness. However unlike in Theorem \ref{theorem:vector-diagonal} for vector enumerators, we do not expect tensor enumerators with off-diagonal terms to factor; the proof of that theorem relied critically on the fact that any code stabilized by a Pauli of the form $P\otimes_j I^{\otimes (n-1)}$ must have a separable factor.

\begin{proposition}
    Let $\mathfrak{C}$ be a $[[n,k]]$ stabilizer code and $J$ be a set of legs of size $m$. Then for $P\not= P'\in \mathcal{P}^m$,
    \begin{enumerate}
        \item $\mathbf{A}^{(J)}(P,P')$ is nonzero if and only if $PP'\otimes_J I^{\otimes(n-m)}$ is a stabilizer;
        \item $\mathbf{B}^{(J)}(P,P')$ is nonzero if and only if $\mathbf{A}^{(J)}(P,P')$ is nonzero.
    \end{enumerate}
\end{proposition}

Exactly as with vector enumerators, the logical legs of the encoding map of a quantum code can be cleaned and therefore its tensor enumerator must be diagonal. This follows from the above proposition. In fact, the components of the enumerator contain the Shor-Laflamme enumerators of the code.

\begin{proposition}
    Let $\mathfrak{C}$ be a $[[n,k]]$ binary quantum code and $\ket{T_\mathfrak{C}}$ the state associated to its encoding map. Let $L$ be the legs of $\ket{T_\mathfrak{C}}$ corresponding to the logical operators. Then
    \begin{align*}
        A(\mathfrak{C}) &= 4^k \cdot \mathbf{A}^{(L)}(I,I;\ket{T_\mathfrak{C}}\bra{T_\mathfrak{C}}), \text{ and}\\
        B(\mathfrak{C}) &= 2^k \cdot \sum_{P\in\mathcal{P}^k} \mathbf{A}^{(L)}(P,P;\ket{T_\mathfrak{C}}\bra{T_\mathfrak{C}}).
    \end{align*}
\end{proposition}

Finally, we state the MacWilliams identity for tensor enumerators. This includes as a special case the MacWilliams identities for scalar and vector enumerators, and is in turn a special case of a general framework for quantum MacWilliams identities covered in the next section. Specifically the following theorem is a special case of Corollary \ref{corollary:tensor-macwilliams} below.

\begin{theorem}\label{theorem:tensor-macwilliams-identity}
    Let $\mathcal{E}$ be an error basis on $\mathfrak{H} = \mathbb{C}^\D$, and $M_1, M_2$ Hermitian operators on $\mathfrak{H}^{\otimes n}$. Let $J \subseteq \{1,\dots, n\}$ be a subset of size $m$. Then
    $$\mathbf{B}^{(J)}(w,z;M_1,M_2) = \Psi\left[\mathbf{A}^{(J)}\left(\tfrac{w+(q^2-1)z}{q},\tfrac{w-z}{q};M_1,M_2\right)\right]$$
    where $\Psi(e_{E,E'}) = \frac{1}{q^{2m}} \sum_{F,F'\in\mathcal{E}^m}\Tr(F^\dagger E F'(E')^\dagger) e_{F,F'}$.
\end{theorem}

Just as we did with vector enumerators, we can use the MacWilliams identity to prove the contravariant behavior of the $\mathbf{B}$ tensor enumerator.

\begin{corollary}
    Let $U = U_1 \otimes \cdots \otimes U_n$ be a local unitary operator on $\mathfrak{H}^{\otimes n}$, and let $J \subseteq \{1, \dots, n\}$ be any subset of legs. Then
    $$\mathbf{B}^{(J)}(z; UM_1 U^\dagger, UM_2 U^\dagger) = \Lambda(U_J^\dagger)\left(\mathbf{B}^{(J)}(z; M_1,M_2)\right),$$
    where $U_J = \bigotimes_{j\in J} U_j$.
\end{corollary}
\begin{proof}
    From the theorem, Theorem \ref{theorem:tensor-A-unitary-covariance}, and Lemma \ref{lemma:Wigner-transform-covariance} we have
    \begin{align*}
        &\mathbf{B}^{(J)}(w,z; UM_1 U^\dagger, UM_2 U^\dagger)\\
        &\quad = \Psi\left[\mathbf{A}^{(J)}\left(\tfrac{w+(q^2-1)z}{q},\tfrac{w-z}{q};UM_1 U^\dagger, UM_2 U^\dagger\right)\right]\\
        &\quad = \Psi\left[\Lambda(U_J)\mathbf{A}^{(J)}\left(\tfrac{w+(q^2-1)z}{q},\tfrac{w-z}{q};M_1, M_2\right)\right]\\
        &\quad = \Lambda(U^\dagger_J) \Psi\left[\mathbf{A}^{(J)}\left(\tfrac{w+(q^2-1)z}{q},\tfrac{w-z}{q};M_1, M_2\right)\right]\\
        &\quad = \Lambda(U_J^\dagger)\left(\mathbf{B}^{(J)}(z; M_1,M_2)\right).
    \end{align*}
\end{proof}

\section{All quantum MacWilliams identities}\label{section:macwilliams}

Let $\mathcal{E}$ be an error basis on $\mathfrak{H}$. 

\begin{definition}
    A \emph{weight function} is defined as any function $\mathrm{wt}:\mathcal{E}^2 \to \mathbb{Z}_{\geq 0}^k \cup \{\bot\}$. A weight function is \emph{scalar} if it is supported on the diagonal: $\mathrm{wt}(E,E') = \bot$ whenever $E \not= E'$. Given a tuple $(\mathrm{wt}_1,\dots, \mathrm{wt}_n)$ of such functions, we define $\mathbf{wt}:\mathcal{E}^{2n} \to \mathbb{Z}_{\geq 0}^k \cup \{\bot\}$ by
    \begin{align*}
        &\mathbf{wt}(E_1 \otimes E'_1 \otimes \cdots \otimes E_n \otimes E'_n)\\
        &\quad = \mathrm{wt}_1(E_1,E'_1) + \cdots + \mathrm{wt}_n(E_n,E'_n),
    \end{align*}
    where we define $t + \bot = \bot$ for any $t\in \mathbb{Z}_{\geq 0}^k \cup \{\bot\}$.
\end{definition}

 For a $k$-tuple of indeterminates $\mathbf{u} = (u_1, \dots, u_k)$ we write
$$\mathbf{u}^{\mathbf{wt}(E,E')} =\left\{\begin{array}{cl} \prod_{j=1}^k u_j^{\mathbf{wt}(E,E')_j} & \text{if $\mathbf{wt}(E,E') \not= \bot$}\\
0 & \text{if $\mathbf{wt}(E,E') = \bot$.}\end{array}\right.$$
Here $\mathbf{wt}(E,E')_j$ is the $j$-th coordinate of $\mathbf{wt}(E,E')$.

 For example, the usual (homogeneous) quantum weight enumerator uses variables $(w,z)$ and has weight function
$$\mathrm{wt}(E,E') = \left\{\begin{array}{cl} (1,0) & \text{if $E = E' = I$}\\ (0,1) & \text{if $E = E' \not= I$}\\ \bot & \text{if $E\not= E'$.}\end{array}\right.$$
The (homogeneous) double enumerator for the usual Pauli group uses variables $(w,x,y,z)$ and has
$$\mathrm{wt}(E,E') = \left\{\begin{array}{cl} 
(1,0,1,0) & \text{if $E = E' = I$}\\ 
(0,1,1,0) & \text{if $E = E' = X$}\\ 
(0,1,0,1) & \text{if $E = E' = Y$}\\ 
(1,0,0,1) & \text{if $E = E' = Z$}\\
\bot & \text{if $E \not= E'$.}
\end{array}\right.$$
The complete enumerator of an error basis $\mathcal{E}$ has variables $\{u_E\}_{E \in \mathcal{E}}$ and the weight $\mathrm{wt}(E,E)$ has a $1$ in the $E$-th position and zero elsewhere, while for $E\not=E'$ the weight 
$\mathrm{wt}(E,E') = \bot$.

For vector enumerators, their weight functions are composed from two weight functions. In all but one coordinate we have a scalar weight function, such as any of those above. In the one remaining coordinate we have variables $\{e_{E,E'}\}_{E,E'\in\mathcal{E}}$ and the weight function where $\mathrm{wt}(E,E')$ equal to $1$ in the $(E,E')$ position and zero elsewhere. Hence, the overall set of variables consists of the scalar variables and $\{e_{E,E'}\}$. The latter set of variables has the feature that every monomial has precisely one of these appearing linearly, hence our interpretation as a vector with polynomial coefficients.

Tensor enumerators are natural extensions of vector enumerators, where in each leg of the tensor we take a \emph{distinct} copy of the weight function $\mathrm{wt}(E,E')$ above. That is for an order $m$ tensor, we have variables $\{e_{E_1,E'_1}, \dots, e_{E_m,E'_m}\}$ and the weight function in the $j$-th leg of the tensor is $\mathrm{wt}(E_j,E'_j)$ that has a $1$ is $(E_j,E'_j)$ position for variable $e_{E_j,E'_j}$ and zero elsewhere. 

We define enumerators of Hermitian operators $M_1,M_2$ over a weight function $\mathbf{wt}$ as
\begin{align*}
    A(\mathbf{u}; M_1, M_2) &= \sum_{E,E' \in \mathcal{E}^n} \Tr(E^\dagger M_1) \Tr(E' M_2) \mathbf{u}^{\mathbf{wt}(E,E')}\\
    B(\mathbf{u}; M_1, M_2) &= \sum_{E,E' \in \mathcal{E}^n} \Tr(E^\dagger M_1 E' M_2) \mathbf{u}^{\mathbf{wt}(E,E')}.
\end{align*}
This recovers scalar, vector, and tensor enumerators over any of weight functions given above as special cases.

Given a weight function $\mathrm{wt}:\mathcal{E}^2 \to \mathbb{Z}_{\geq 0}^k\cup\{\bot\}$, consider the function $f(E,E') = \mathbf{u}^{\mathrm{wt}(E,E')}$. We define a form of generalized discrete Wigner transform of this function as
\begin{align*}
    \hat{f}(D,D') &= \frac{1}{q^2}\sum_{E,E'} \Tr(E^\dagger D E' (D')^\dagger) f(E,E')\\
    &= \frac{1}{q^2}\sum_{E,E'} \Tr(E^\dagger D E' (D')^\dagger) \mathbf{u}^{\mathrm{wt}(E,E')}.
\end{align*}

\begin{theorem}
    Let $\mathcal{E}$ be an error basis on $\mathfrak{H}$, and let $\mathbf{wt}$ be associated to weight functions $(\mathrm{wt}_1, \dots, \mathrm{wt}_n)$, where each $\mathrm{wt}_j:\mathcal{E}^2 \to \mathbb{Z}_{\geq 0}^k \cup \{\bot\}$. Suppose there exists an algebraic mapping $\Phi(\mathbf{u}) = (\Phi_1(\mathbf{u}), \dots, \Phi_k(\mathbf{u}))$ where each $j=1, \dots, n$ has:
    \begin{equation}\label{equation:condition}
    \Phi(\mathbf{u})^{\mathrm{wt}_j(D,D')} = \tfrac{1}{q^2}\sum_{E,E'} \Tr(E^\dagger D E' (D')^\dagger) \mathbf{u}^{\mathrm{wt}_j(E,E')}.
    \end{equation}
    Then the enumerators for this weight function satisfy
    $$B(\mathbf{u};M_1,M_2) = A(\Phi(\mathbf{u});M_1, M_2).$$
\end{theorem}
\begin{proof}
    We begin by expressing
    \begin{align*}
        M_1 &= \frac{1}{q^n} \sum_{D \in \mathcal{E}^n} \Tr(D^\dagger M_1) D, \\
        M_2 &= \frac{1}{q^n} \sum_{D' \in \mathcal{E}^n} \Tr(D' M_2) (D')^\dagger.
    \end{align*}
    Then
    \begin{align*}
        &B(\mathbf{u};M_1,M_2) = \sum_{E,E' \in \mathcal{E}^n} \Tr(E^\dagger M_1 E' M_2) \mathbf{u}^{\mathbf{wt}(E,E')}\\
        &\quad= \frac{1}{q^{2n}} \sum_{D,D',E,E' \in \mathcal{E}^n} \left(\Tr(D^\dagger M_1)\Tr(D' M_2)\right.\\
        &\qquad\qquad\qquad\qquad\qquad \left. \cdot \Tr(E^\dagger D E' (D')^\dagger) \mathbf{u}^{\mathbf{wt}(E)}\right).
    \end{align*}
    Now,
    \begin{align*}
        &\frac{1}{q^{2n}} \sum_{E,E'\in\mathcal{E}^n} \Tr(E^\dagger DE' (D')^\dagger) \mathbf{u}^{\mathbf{wt}(E,E')}\\
        &\qquad = \frac{1}{q^{2n}} \sum_{E,E'\in\mathcal{E}^n} \prod_{j=1}^n \Tr(E_j^\dagger D_j E'_j (D'_j)^\dagger) \mathbf{u}^{\mathrm{wt}_j(E_j,E'_j)}\\
        &\qquad = \prod_{j=1}^n \frac{1}{q^2} \sum_{E,E'\in \mathcal{E}} \Tr(E^\dagger D_j E' (D'_j)^\dagger) \mathbf{u}^{\mathrm{wt}_j(E,E')}\\
        &\qquad = \prod_{j=1}^n \Phi(\mathbf{u})^{\mathrm{wt}_j(D_j,D'_j)} = \Phi(\mathbf{u})^{\mathbf{wt}(D,D')}.
    \end{align*}
    Therefore,
    \begin{align*}
    &B(\mathbf{u};M_1,M_2)\\
    &\qquad = \sum_{D,D' \in \mathcal{E}^n} \Tr(D^\dagger M_1)\Tr(D' M_2) \Phi(\mathbf{u})^{\mathbf{wt}(D,D')}\\
    &\qquad = A(\Phi(\mathbf{u});M_1, M_2).
    \end{align*}
\end{proof}

So to obtain a quantum MacWilliams identity, we need to provide an algebraic map that satisfies condition (\ref{equation:condition}) for each tensor factor. In all of the cases we examine below, the form of this map is determined by simplifying the generalized discrete Wigner transform when applied to the weight function monomial.

\begin{corollary}
    Let $\mathcal{E}$ be any error basis on $\mathfrak{H} = \mathbb{C}^\D$, variables $\mathbf{u} = \{u_E\}_{E\in\mathcal{E}}$, and the complete weight function $\mathrm{wt}$ given implicitly by $\mathbf{u}^{\mathrm{wt}(E,E')} = u_E$ when $E = E'$ and zero when $E \not= E'$. Let $A$ and $B$ be the complete weight enumerators associated to this weight function. Then
    $$B(\{u_D\};M_1,M_2) = A\left(\frac{1}{q}\sum_{E\in\mathcal{E}}\omega(D,E)u_E; M_1, M_2\right).$$
\end{corollary}
\begin{proof}
    For any scalar weight function
    \begin{align*}
        &\frac{1}{q^2}\sum_{E,E'}\Tr(E^\dagger D E' (D')^\dagger) \mathbf{u}^{\mathrm{wt}(E,E')}\\
        &\quad = \frac{1}{q^2}\sum_{E}\Tr(E^\dagger D E (D')^\dagger) \mathbf{u}^{\mathrm{wt}(E,E)}\\
        &\quad = \frac{1}{q^2}\sum_{E} \omega(D,E) \Tr(D(D')^\dagger) \mathbf{u}^{\mathrm{wt}(E,E)}\\
        &\quad = \left\{\begin{array}{cl} \tfrac{1}{q} \sum_{E} \omega(D,E) \mathbf{u}^{\mathrm{wt}(E,E)} & \text{if $D=D'$}\\
        0 & \text{if $D\not= D'$.}\end{array}\right.
    \end{align*}
    Hence for $D = D'$ we take
    $$\Phi(\mathbf{u})^{\mathrm{wt}(D,D)} = \Phi_D(\mathbf{u}) = \frac{1}{q} \sum_{E} \omega(D,E) u_E,$$
    which then satisfies condition (\ref{equation:condition}) of the theorem. When $D \not= D'$ we have $0=0$, which is always satisfied.
\end{proof}

\begin{corollary}[\cite{shor1997quantum},\cite{rains1998quantum}]\label{corollary:scalar-macwilliams}
For any error basis $\mathcal{E}$ on $\mathfrak{H} = \mathbb{C}^\D$ the Shor-Laflamme enumerators satisfy
$$B(w,z;M_1,M_2) = A\left(\tfrac{w+(q^2-1)z}{q},\tfrac{w-z}{q};M_1,M_2\right).$$
\end{corollary}
\begin{proof}
    We can reduce from the complete enumerator of the previous result by setting $u_I = w$ and $u_D = z$ for $D\not= I$. So in the case $D = D' = I$ we have
    $$\Phi_I(w,z) = \frac{1}{q}\sum_{E\in\mathcal{E}} u_E = \frac{w + (q^2-1)z}{q}.$$
    While for any $D = D' \not= I$ we have
    \begin{align*}
        \Phi_D(w,z) &= \frac{1}{q} \sum_E \omega(D,E)u_E\\
        &= \frac{1}{q}\left(w - z + z\cdot\sum_{E} \omega(D,E)\right) = \frac{w - z}{q}.
    \end{align*}
    All $\Phi_D(w,z)$ are equal, so $\Phi(w,z) = (\Phi_I(w,z),\Phi_D(w,z))$ is well defined and satisfies condition (\ref{equation:condition}).
\end{proof}

The double weight enumerator of Hu, Yang, and Yau, for both the binary \cite{hu2019complete} and non-binary \cite{hu2020weight} Pauli bases,
are defined over homogeneous variables $(w,x,y,z)$. When $\D > 2$, we can define a refined double weight enumerator over variables $\mathbf{u} = (\mathbf{x},\mathbf{z}) = (x_0, \cdots x_{\D-1},z_0, \cdots, z_{\D-1})$ using the scalar weight function implicitly defined via
$$(\mathbf{x},\mathbf{z})^{\mathrm{wt}(X^a Z^b,X^a Z^b)} = x_a z_b,$$
and zero otherwise. 

\begin{corollary}
    For the Pauli basis $\mathcal{P}$ on $\mathfrak{H} = \mathbb{C}^\D$, the refined double quantum weight enumerators defined by the weight function above satisfy
    \begin{align*}
        &B(x_a,z_b;M_1,M_2)\\
        &\qquad = A(\tfrac{1}{\sqrt{q}} \sum_{d} \zeta^{ad} z_d, \tfrac{1}{\sqrt{q}} \sum_c \zeta^{-bc} x_c; M_1,M_2),
    \end{align*}
    where $\zeta = e^{2\pi i/\D}$.
\end{corollary}
\begin{proof}
    Write $\Phi(\mathbf{x},\mathbf{z}) = (\Theta_a(\mathbf{x},\mathbf{z}),\Psi_b(\mathbf{x},\mathbf{z}))$. When $D = D' = X^a Z^b$ we require
    \begin{align*}
        &(\Phi(\mathbf{x},\mathbf{z}))^{\mathrm{wt}(X^aZ^b,X^a Z^b)} = \Theta_a(\mathbf{x},\mathbf{z})\cdot \Psi_b(\mathbf{x},\mathbf{z})\\
        &\qquad = \frac{1}{q} \sum_{c,d} \zeta^{ad - bc} (\mathbf{x},\mathbf{z})^{\mathrm{wt}(X^c Z^d,X^c Z^d)} \\        
        &\qquad = \frac{1}{q} \sum_{c,d} \zeta^{ad - bc} x_c z_d \\
        &\qquad = \left(\frac{1}{\sqrt{q}} \sum_c \zeta^{-bc} x_c\right) \left(\frac{1}{\sqrt{q}} \sum_{d} \zeta^{ad} z_d\right).
    \end{align*}
    Therefore, the condition (\ref{equation:condition}) is satisfied for
    $$\Theta_a(\mathbf{x},\mathbf{z}) = \frac{1}{\sqrt{q}} \sum_{d} \zeta^{ad} z_d \text{ and } \Psi_b(\mathbf{x},\mathbf{z}) = \frac{1}{\sqrt{q}} \sum_c \zeta^{-bc} x_c.$$
    When $D \not= D'$ we need $0=0$, which always holds.
\end{proof}

\begin{corollary}[\cite{hu2020weight}]\label{corollary:double-macwilliams}
    For the Pauli basis $\mathcal{P}$ on $\mathfrak{H} = \mathbb{C}^\D$, then the double quantum weight enumerators of Hu, Yang, and Yau satisfy:
    \begin{align*}
        &D(w,x,y,z;M_1,M_2)\\
        &\qquad = C\left(\tfrac{y+(\D-1)x}{\sqrt{\D}}, \tfrac{w-z}{\sqrt{\D}}, \tfrac{w+(\D-1)}{\sqrt{\D}}, \tfrac{y-x}{\sqrt{\D}}; M_1, M_2\right)
    \end{align*}
\end{corollary}
\begin{proof}
    We can reduce from the refined double enumerator above by setting $x_0 = y$ and $x_a = x$ for $a > 0$, and $z_0 = w$ and $z_b = z$ for $b > 0$. Then
    \begin{align*}
        \Theta_0(w,x,y,z) &= \tfrac{1}{\sqrt{\D}}(w + (\D-1)z)\\
        \Theta_a(w,x,y,z) &= \tfrac{1}{\sqrt{\D}}(w - z + z\sum_{d} \zeta^{ad})\\
        &= \tfrac{1}{\sqrt{\D}}(w - z) \text{ for any $a > 0$}\\
        \Psi_0(w,x,y,z) &= \tfrac{1}{\sqrt{\D}}(y + (\D-1)x)\\
        \Psi_b(w,x,y,z) &= \tfrac{1}{\sqrt{\D}}(y - x + x\sum_{d} \zeta^{ad})\\
        &= \tfrac{1}{\sqrt{\D}}(y - x) \text{ for any $b > 0$.}
    \end{align*}
    Since $\Theta_a$ and $\Psi_b$ remain the same for any $a,b > 0$, any choice of these provide a well defined transform that satisfies (\ref{equation:condition}).
\end{proof}

\begin{corollary}[\cite{hu2020weight}]\label{corollary:complete-macwilliams}
    For the Pauli basis $\mathcal{P}$ on $\mathfrak{H} = \mathbb{C}^\D$, the complete quantum enumerators of Hu, Yang, and Yau satisfy
    $$F(u_{ab};M_1,M_2) = E(\tfrac{1}{q} \sum_{c,d} \zeta^{ad-bc} u_{cd};M_1,M_2)$$
    where $\zeta = e^{2\pi i/\D}$.
\end{corollary}
\begin{proof}
    The complete weight function is essentially the identity, hence for $D = D' = X^a Z^b$ we require we require
    \begin{align*}
        &\Phi(\mathbf{u})^{\mathrm{wt}(X^a Z^b, X^a Z^b)} = \Phi_{(a,b)}(\mathbf{u})\\
        &\qquad = \frac{1}{q} \sum_{c,d} \zeta^{ad-bc} \mathbf{u}^{\mathrm{wt}(X^c X^d, X^c x^d)} = \frac{1}{q} \sum_{c,d} \zeta^{ad-bc} u_{cd}.
    \end{align*}
    When $D \not= D'$ we need $0=0$, which is always satisfied.
\end{proof}

Finally we consider the vector and tensor enumerators over any scalar system that satisfies a MacWilliams identity, such as any of those above. 

\begin{corollary}\label{corollary:tensor-macwilliams}
    Let $\mathcal{E}$ be an error basis on $\mathfrak{H} = \mathbb{C}^\D$ and $\widetilde{\mathrm{wt}}$ be any scalar weight function that satisfies condition (\ref{equation:condition}) of the theorem for algebraic map $\Phi(\mathbf{u})$. Let $J = \{j_1, \dots, j_m\} \subseteq \{1, \dots, n\}$ and define $\mathbf{wt}$ on $(\mathcal{E}\times \mathcal{E})^n$ implicitly by
    \begin{align*}
        &(\mathbf{u},\mathbf{e}_1, \dots, \mathbf{e}_m)^{\mathbf{wt}(E,E')} \\
        &\quad = \left(\prod_{k\not\in J} \mathbf{u}^{\widetilde{\mathrm{wt}}(E_k,E'_k)}\right)e_{E_{j_1},E'_{j_1}} \cdots e_{E_{j_m},E'_{j_m}},
    \end{align*}
    where each set of variables $\mathbf{e}_j = \{e_{E,E'}\}_{E,E'\in\mathcal{E}}$. Define enumerators $\mathbf{A}^{(J)}$ and $\mathbf{B}^{(J)}$ by this weight function. Then for any Hermitian operators $M_1,M_2$ on $\mathfrak{H}^{\otimes n}$,
    $$\mathbf{B}^{(J)}(\mathbf{u}; M_1,M_2) = \Psi[\mathbf{A}^{(J)}(\Phi(\mathbf{u}); M_1,M_2)]$$
    where $\Psi(e_{E,E'}) = \frac{1}{q^{2m}} \sum_{F,F'\in\mathcal{E}^m} \Tr(F^\dagger E F' (E')^\dagger)e_{F,F'}$.
\end{corollary}
\begin{proof}
    This is immediate. For the legs not in $J$, condition (\ref{equation:condition}) of the theorem is satisfied by assumption. For the legs of $J$, we have $\Psi(e_{E,E'}) = \frac{1}{q^{2m}} \sum_{F,F'\in\mathcal{E}^m} \Tr(F^\dagger E F' (E')^\dagger)e_{F,F'}$ is condition (\ref{equation:condition}) of the theorem.
\end{proof}

\section{Tensor networks}\label{section:networks}

\begin{figure}[b]
    \begin{center}
        \includegraphics[scale=0.75]{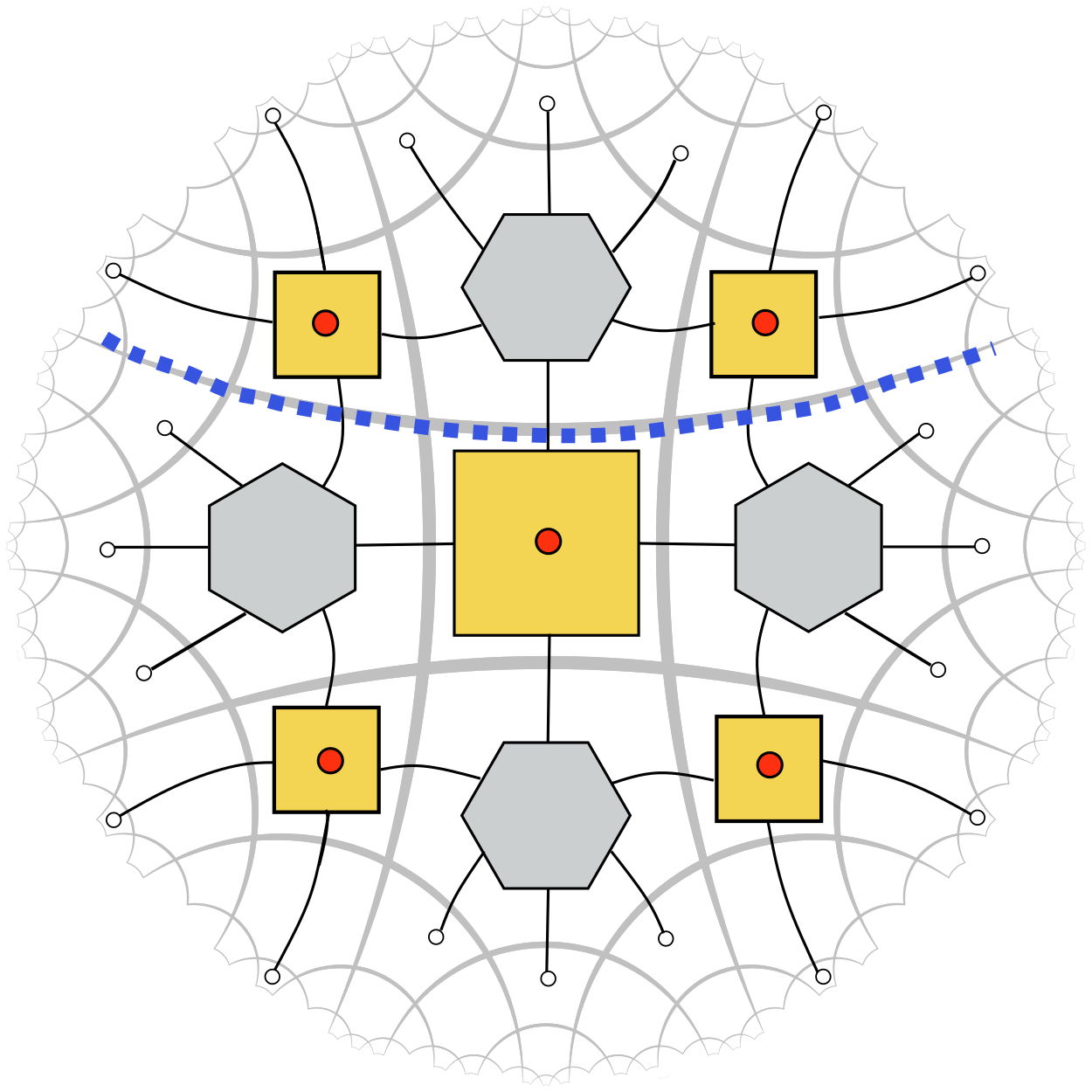}
        \caption{Holographic code from \cite{ABSC} with a cut line. Only three legs cross the cut, hence the weight enumerators of this code can be computed by tracing two tensor enumerators with three legs.}
    \end{center}
\end{figure}

Given a tensor network state expressed in the computational basis of $(\mathbb{C}^q)^{\otimes n}$,
$$\ket{T} = \sum_{j_1,\cdots, j_n} c_{j_1,\cdots, j_n} \ket{j_1, \dots, j_n},$$
we define its trace on, say, legs $1$ and $2$ to be
$$\wedge^{1,2}\ket{T} \propto \sum_j \sum_{j_3,\cdots, j_n} c_{j,j,j_3\cdots, j_n} \ket{j_3, \dots, j_n}.$$
Note that we may need to renormalize to get a state of norm one. Writing $\ket{\beta} = \frac{1}{\sqrt{q}}\sum_j\ket{jj}$,
we can informally write
$$\wedge^{1,2}\ket{T} = (\bra{\beta} \otimes I^{\otimes n-2}) \ket{T}$$
or more generally $\wedge^{j,k}\ket{T} = (\bra{\beta} \otimes_{j,k} I^{\otimes n-2}) \ket{T}$ using the notation of the previous sections.

For a general Hermitian operator $M$ on $(\mathbb{C}^q)^{\otimes n}$, and legs $1 \leq j < k \leq n$ implicitly define $\wedge^{j,k} M$ on $(\mathbb{C}^q)^{\otimes (n-2)}$ by
$$\Tr(A(\wedge^{j,k} M)) = \Tr((\ket{\beta}\bra{\beta} \otimes_{j,k} A) M)$$
for any operator $A$ on $(\mathbb{C}^q)^{\otimes (n-2)}$.

Now consider any error basis $\mathcal{E}$ on $\mathfrak{H} = \mathbb{C}^\D$. Note that for any $E,F$ we have that $E\otimes E^*$ commutes with $F \otimes F^*$ since $EF = \omega(E,F) FE$ implies that $E^*F^* = \omega(E,F)^{-1} F^* E^*$ and so
\begin{align*}
    &(E\otimes E^*)(F\otimes F^*)\\
    &\qquad = \omega(E,F)\omega(E,F)^{-1} (F\otimes F^*)(E\otimes E^*).
\end{align*}

\begin{lemma}
    For any error basis $\mathcal{E}$ on $\mathfrak{H} = \mathbb{C}^\D$ we have
    $$\ket\beta\bra\beta = \frac{1}{\D^2} \sum_{E\in\mathcal{E}} E \otimes E^*.$$
\end{lemma}
\begin{proof}
    As $\ket\beta\bra\beta$ is an operator on $\mathfrak{H}\otimes\mathfrak{H}$ we can write
    $$\ket\beta\bra\beta = \sum_{E,F} c_{E,F} E \otimes F.$$
    But,
    \begin{align*}
        c_{E,F} &= \tfrac{1}{\D^2} \Tr(\ket\beta\bra\beta (E\otimes F)^\dagger)\\
        &= \tfrac{1}{\D^3} \sum_{j,k} \bra{j,j} E^\dagger \otimes F^\dagger \ket{k,k}\\
        &= \tfrac{1}{\D^3} \sum_{j,k} \bra{k} E^* \ket{j} \bra{j} F^\dagger \ket{k}\\
        &= \tfrac{1}{\D^3} \Tr(E^* F^\dagger) = \left\{\begin{array}{cl} \frac{1}{\D^2} & \text{ if $F = E^*$}\\ 0 & \text{otherwise.}\end{array}\right.
    \end{align*}
\end{proof}

Recall that tensor enumerators lie in 
$$\mathbb{R}[z]\otimes_\mathbb{R} V^{\otimes m} = \mathrm{span}_{\mathbb{R}[z]}\{ e_{E,E'} \::\: E,E'\in \mathcal{E}^{\otimes m}\}.$$
We define the trace on legs $j,k$ on this space as
$$\wedge^{j,k} e_{E,E'} = e_{E\setminus\{E_j,E_k\}, E'\setminus\{E'_j,E'_k\}}$$
if $E_k = E_j^*$ and $E'_k = (E'_j)^*$, and zero otherwise.
We extend linearly to all of $\mathbb{R}[z]\otimes_\mathbb{R} V^{\otimes m}$.

\begin{theorem}\label{theorem:enumerator-trace}
    Suppose $j,k \in J \subseteq \{1,\dots,m\}$. Then
    $$\wedge^{jk}\mathbf{A}^{(J)}(z;M_1,M_2) = \mathbf{A}^{(J\setminus\{j,k\})}(z;\wedge^{j,k}M_1,\wedge^{j,k}M_2),$$
    and similarly for $\mathbf{B}^{(J)}$.
\end{theorem}
\begin{proof}
    For $\textbf{A}$ we directly compute,
    \begin{align*}
        &\wedge^{jk}\mathbf{A}^{(J)}(z;M_1,M_2)\\
        &\ = \sum_{E,E',F} \left[\Tr((E\otimes_J F)^\dagger M_1)\Tr((E'\otimes_J F) M_2\right]\\
        &\qquad \qquad \cdot z^{\mathrm{wt}(F)}\left[\wedge^{j,k}e_{E,\bar{E}}\right]\\
        &\ = \sum_F\sum_{\tilde{E},\tilde{E}'} \sum_{G,G'} \left\{ \Tr([((G\otimes G^*) \otimes_{j,k} \tilde{E}) \otimes_J F]^\dagger M_1) \right.\\
        &\qquad \qquad \cdot \left. \Tr([((G'\otimes (G')^*) \otimes_{j,k} \tilde{E}') \otimes_J F] M_2)\right\}\\
        &\qquad \qquad \cdot z^{\mathrm{wt}(F)} e_{\tilde{E},\tilde{E}'}\\
        &\ = \sum_{\tilde{E},\tilde{E}',F} \Tr((\tilde{E}\otimes_{J\setminus\{j,k\}} F)^\dagger (\wedge^{j,k} M_1))\\
        &\qquad \qquad \cdot \Tr((\tilde{E}'\otimes_{J\setminus\{j,k\}} F) (\wedge^{j,k} M_2)) z^{\mathrm{wt}(F)} e_{\tilde{E},\tilde{E}'}\\
        &\ = \mathbf{A}^{(J\setminus\{j,k\})}(z; \wedge^{jk}M_1, \wedge^{jk}M_2)
    \end{align*}
where we write $\tilde{E} = E\setminus\{E_j,E_k\}$ and similarly for $\tilde{E}'$.

For $\textbf{B}$ we appeal to the MacWilliams identity. Recall the generalized discrete Wigner transform:
$$\Psi^{(m)}(e_{E,E'}) = \frac{1}{q^{2m}} \sum_{F,F'}\Tr(F^\dagger E F' (E')^\dagger) e_{F,F'}.$$
Here we decorate the transform as we will be working in multiple degrees. We have
\begin{align*}
    &\wedge^{j,k}[\Psi^{(m)}(e_{E,E'})] = \frac{1}{q^{2m}} \sum_{F,F'}\Tr(F^\dagger E F' (E')^\dagger) (\wedge^{j,k} e_{F,F'})\\
    &\qquad = \frac{1}{\D^{2m}}\sum_{F,F'} \Tr(F^\dagger E_j F' (E'_j)^\dagger) \Tr(F^\mathrm{tr} E_k (F')^* (E'_k)^\dagger)\\
    & \qquad \qquad \qquad \cdot \sum_{\tilde{F},\tilde{F}'} \Tr(\tilde{F}^\dagger \tilde{E} \tilde{F}' (\tilde{E}')^\dagger) e_{\tilde{F},\tilde{F}'}\\
    &\qquad = \frac{1}{\D^4} \sum_{F,F'} \Tr(F^\dagger E_j F' (E'_j)^\dagger) \Tr(F^\mathrm{tr} E_k (F')^* (E'_k)^\dagger)\\
    & \qquad \qquad \qquad \cdot \Psi^{(m-2)}(e_{\tilde{E},\tilde{E}'}),
\end{align*}
where like above we have written $\tilde{E} = E\setminus\{E_j,E_k\}$, and similarly for $F,E',F'$. 

Now,
\begin{align*}
    &\sum_{F,F'} \Tr(F^\dagger E_j F' (E'_j)^\dagger) \Tr(F^\mathrm{tr} E_k (F')^* (E'_k)^\dagger)\\
    &\quad = \sum_{F,F'} \Tr((F\otimes F^*)^\dagger (E_j\otimes E_k) (F'\otimes (F')^*) (E'_j \otimes E'_k)^\dagger)\\
    &\quad = \D^4 \Tr(\ket\beta\bra\beta (E_j\otimes E_k) \ket\beta\bra\beta (E'_j \otimes E'_k)^\dagger)\\
    &\quad = \D^4 \bra\beta E_j\otimes E_k\ket\beta \bra\beta E'_j \otimes E'_k \ket\beta.
\end{align*}
But just as we have above
\begin{align*}
    \bra\beta E_j\otimes E_k\ket\beta &= \tfrac{1}{\D} \sum_{r,s} \bra{r,r} E_j \otimes E_k \ket{s,s}\\
    &= \tfrac{1}{\D} \sum_{r,s} \bra{s} E_j^\mathrm{tr} \ket{r}\bra{r} E_k \ket{s}\\
    &= \tfrac{1}{\D}\Tr(E_j^\mathrm{tr} E_k) = \left\{\begin{array}{cl} 1 & \text{ if $E_k = E_j^*$}\\ 0 & \text{otherwise.}\end{array}\right.
\end{align*}
Thus, $e_{\tilde{E},\tilde{E}'} = \wedge^{j,k}(e_{E,E})$ and
\begin{align*}
    \wedge^{j,k}[\Psi^{(m)}(e_{E,E'})] &= \Psi^{(m-2)}(e_{\tilde{E},\tilde{E}'})\\
    &= \Psi^{(m-2)}(\wedge^{j,k}(e_{E,E'})).
\end{align*}

Therefore,
\begin{align*}
    &\wedge^{j,k}\mathbf{B}^{(J)}(w,z;M_1,M_2)\\
    &\ = \wedge^{j,k}\Psi^{(m)}\left[\mathbf{A}^{(J)}(\tfrac{w+(q^2-1)z}{q},\tfrac{w-z}{q};M_1,M_2)\right]\\
    &\ = \Psi^{(m-2)}\left[\mathbf{A}^{(J\setminus\{j,k\})}(\tfrac{w+(q^2-1)z}{q},\tfrac{w-z}{q};\wedge^{j,k}M_1,\wedge^{j,k}M_2)\right]\\
    &\ = \mathbf{B}^{(J\setminus\{j,k\})}(w,z;\wedge^{j,k}M_1,\wedge^{j,k}M_2).
\end{align*}
\end{proof}

\begin{figure*}        
    \begin{center}
        \includegraphics[width=0.95\linewidth]{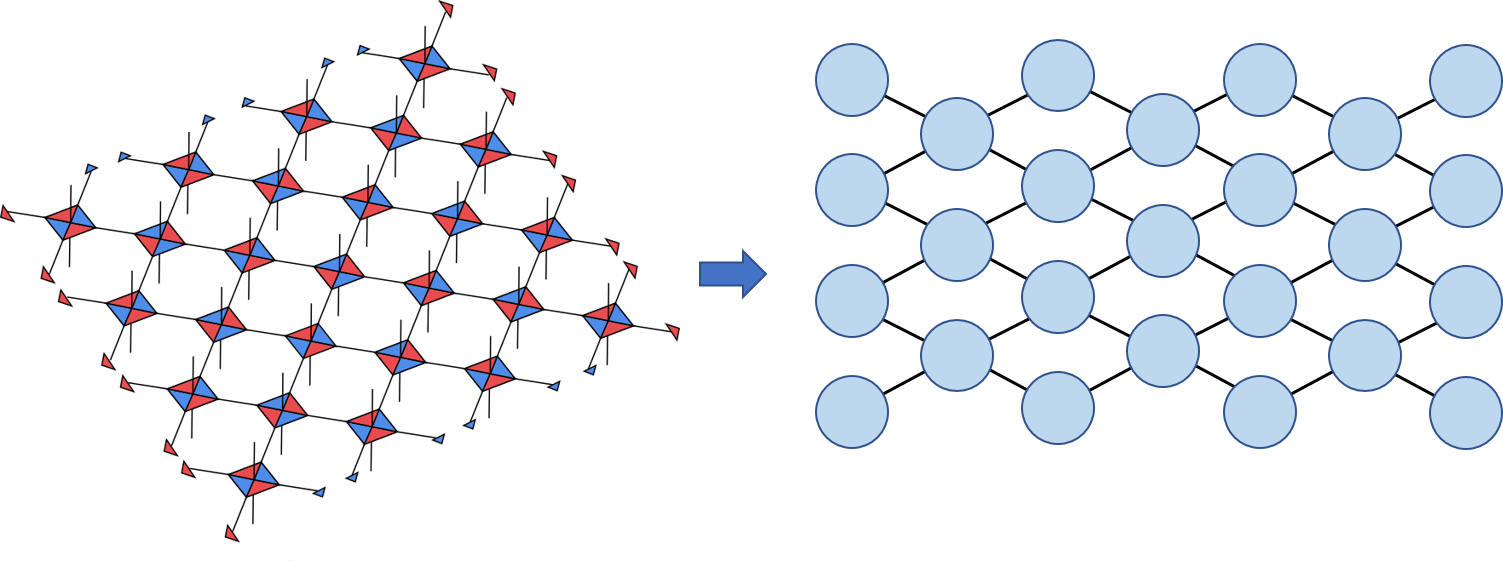}
        \caption{Left: the encoding map of a $[[25,1,4]]$ surface code written as a tensor network. Right: contraction of tensor enumerators. Each tensor $\mathbf{A}^{(J)}(z)$ is a tensor enumerator constructed from the corresponding lego on the left.}
    \end{center}
\end{figure*}

\begin{example}
Consider the holographic code from \cite{ABSC} as shown in Figure~1 above. This five logical qubit shallow holographic code only has 20 physical qubits, and so it would not be numerically challenging to compute its enumerators directly from the definitions. However, it is an informative example that we can illustrate much more easily using tensor tracing. Namely, there are only three legs across the blue cut in the figure and hence we can compute the three index tensor enumerator above the cut separately. For example, above the cut the left Bacon-Shor code (say in the $Y$-gauge) has tensor enumerator 
\begin{align*}
    \mathbf{A}^{(3,4)}_{[[4,1,2]]} &= e_{I,I} + z^2 e_{I,X} + z^2 e_{X,I} + z^2 e_{X,X}\\
    & \qquad + z^2 e_{Z,Z} + z e_{Z,Y} + z e_{Y,Z} + z^2 e_{Y,Y}.
\end{align*}
Note that the two legs on the boundary will never be traced, and so we need only track the tensor for remaining two legs, legs 3 and 4 by our counting. The tensor enumerator for the perfect tensor with legs 4, 5, and 6 has 64 nonzero terms and so we will not write out $\mathbf{A}^{(4,5,6)}$ explicitly, however each coefficient is simply a single monomial. Tracing leg 3 of $\mathbf{A}^{(3,4)}_{[[4,1,2]]}$ with leg 6 of $\mathbf{A}^{(4,5,6)}$ produces a new tensor enumerator with 3 legs and 64 terms. Tracing the result of this with the tensor enumerator of the right Bacon-Shor code again produces a tensor enumerator with 3 legs and 64 terms.

When we trace in the perfect tensor on the left (or right) of lower wedge we will expand to four legs across the cut, and therefore will need to store a tensor enumerator with 256 terms. But this is the largest cut needed and so the enumerator of the whole code can be computing by manipulating tensors consisting of 256 enumerators. In some sense, this is optimal as tensor enumerator of the middle Bacon-Shor code has 4 legs itself, albeit only eight of its 256 terms are nonzero.
\end{example}

\begin{example}
    Recall that the surface code can be built from contracting encoding isometries of the $[[5,1,2]]$ codes \cite{cao2022quantum}. Its encoding map has a non-trivial kernel and can be represented graphically by Figure~2 (left). We take the upward pointing legs as inputs while the downward pointing ones as outputs (physical qubits). The blue and red rank-1 tensors are $|+\rangle$ and $|0\rangle$ respectively. 
    
    By taking each lego on the left diagram, we may construct its corresponding tensor enumerator polynomial, represented by a blue tensor in Figure~2 (right). The network contraction yields the overall scalar enumerator of the entire code. 
    
    The examples we construct so far have a small number qubits such that their enumerators can also be computed by brute force. However, we should keep in mind that the tensor network method is far more scalable \cite{followup}. For example, the surface code tensor network may be easily extended to larger sizes. In Fig.~3, we obtain the coefficients of the double enumerator $C(x,z)$ of a surface code on a strip with $748$ qubits. Recall that it captures the $X,Z$ weight distribution of the stabilizer elements. To contrast with the brute force method, computing the same enumerator polynomial will involve checking the weights of at least $2^{747}$ terms, which is clearly infeasible. 
\end{example}

\begin{figure}
    \centering
    \includegraphics[width=\linewidth]{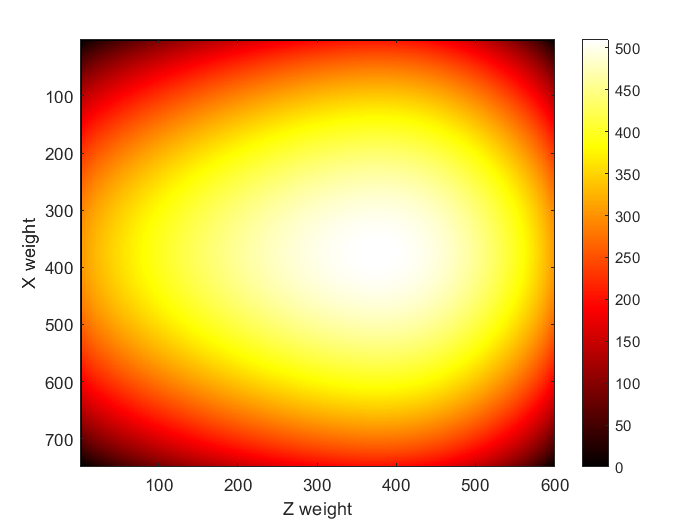}
    \caption{Plot of $\log C[d_x,d_z]$, where $C[d_x,d_z]$ are the coefficients of the double enumerator of a 3-by-150 surface code. $d_x,d_z$ are the $X$ and $Z$ weights respectively.}
\end{figure}

\end{document}